\documentclass{amsart}
\usepackage{amsmath, datetime}
\usepackage{amssymb,amsmath,amscd,epsfig,amsthm,import, comment,xcolor, enumerate}

\definecolor{darkgreen}{rgb}{0,0.5,0}
\definecolor{darkred}{rgb}{0.7,0,0}
\definecolor{darkblue}{rgb}{0,.2,.7}

\usepackage[colorlinks, 
citecolor=darkred, linkcolor=darkgreen, urlcolor=darkblue,
 pdfpagelabels=true,
 unicode=true,
 hyperfootnotes=false,
]{hyperref}

\hypersetup{
 pdfauthor={Nadine Gro{\ss}e, Alejandro Uribe and Hanne Van den Bosch},
 pdftitle={Local Boundary Conditions for Dirac-type operators},
 pdflang={en}
}

\newtheorem{theorem}{Theorem}[section]
\newtheorem{proposition}[theorem]{Proposition}
\newtheorem{lemma}[theorem]{Lemma}
\newtheorem{corollary}[theorem]{Corollary}
\theoremstyle{definition}
\newtheorem{definition}[theorem]{Definition}
\newtheorem{example}[theorem]{Example}

\theoremstyle{remark}
\newtheorem{remark}[theorem]{Remark}

\newcommand{\tps}[2]{\texorpdfstring{$#1$}{#2}}
\numberwithin{equation}{section}
\newcommand{\be}{\begin{equation}}
\newcommand{\ee}{\end{equation}}

\newcommand{\bbC}{{\mathbb C}}

\newcommand{\bbR}{{\mathbb R}}

\newcommand{\bbS}{{\mathbb S}}

\newcommand{\calL}{{\mathcal L}}

\newcommand{\calE}{{\mathcal E}}

\newcommand{\calC}{{\mathcal C}}
\newcommand{\calR}{{\mathcal R}}

\newcommand{\calS}{{\mathcal S}}

\newcommand{\frakb}{{\mathfrak b}}

\newcommand{\romd}{{\mathrm d}}

\newcommand{\sfa}{\mathsf{a}}
\newcommand{\sfd}{\mathsf{d}}

\DeclareMathOperator{\End}{End}
\DeclareMathOperator{\rank}{rank}
\DeclareMathOperator{\Ker}{Ker}
\DeclareMathOperator{\Ran}{Ran}

\newcommand{\inner}[2]{\langle#1,#2\rangle}

\newcommand{\norm}[1]{\lVert#1\rVert}

\newcommand{\R}{\mathbb{R}}
\newcommand{\C}{\mathbb{C}}
\newcommand{\N}{\mathbb{N}}

\newcommand{\di}{\, \mathrm d}

\newcommand{\Id}{\mathrm{Id}}

\newcommand{\dom}{\operatorname{dom}}
\DeclareMathOperator{\Gr}{\operatorname{Gr}}

\definecolor{darkblue}{rgb}{0,.2,.7}
\definecolor{darkred}{rgb}{0.7,0,0}

\newcommand{\ii}{\mathrm{i}}

\usepackage[normalem]{ulem}

\newcommand{\define}{\mathrm{:=}}

\definecolor{darkgreen}{RGB}{27, 94, 32}

\title{Local Boundary Conditions for Dirac-type operators}
\author{N. Gro\ss e} 
\address{Mathematical Institute, University of Freiburg, Germany }
\email{nadine.grosse@math.uni-freiburg.de}

\author{A. Uribe}
\address[Corresponding auhtor]{Mathematics Department\\
	University of Michigan\\Ann Arbor, Michigan 48109, U.S.A.}
\email{uribe@umich.edu}

\author{H. Van Den Bosch}
\address{Department of Mathematical Engineering and Center for Mathematical Modeling (CNRS IRL2807)\\ Universidad de Chile, Santiago, Chile. }
\email{hvdbosch@dim.uchile.cl}

\makeatletter
\@namedef{subjclassname@2020}{%
	\textup{2020} Mathematics Subject Classification}
\makeatother

\subjclass[2020]{Primary 81Q10; Secondary 35G15, 53C27.}

\keywords{Dirac operator, local boundary condition, self-adjointess, Shapiro-Lopatinski}    
    
\begin{document}

\begin{abstract}
    We consider Dirac-type operators on manifolds with boundary, and 
set out to determine all local smooth boundary conditions that give rise to 
(strongly) regular self-adjoint operators. By combining the general theory of boundary value problems for Dirac operators as in \cite{BB12} and pointwise considerations, for local smooth boundary conditions the question of being self-adjoint resp. regular is fully translated into linear-algebraic language at each boundary point. We analyse these conditions and classify them
in low dimensions and ranks. In particular, we classify all local self-adjoint regular boundary conditions for Dirac spinors (four spinor components) in dimensions $3$ and $4$. With the same techniques we can also treat transmission boundary conditions.

\end{abstract}
\maketitle
	
	\tableofcontents

 \section{Introduction}

The goal of this paper is to present a systematic study of local smooth boundary conditions
that give rise to self-adjoint Dirac operators
and that satisfy an elliptic regularity condition.
Symmetric local smooth boundary conditions ensure that the current perpendicular to the boundary vanishes pointwise on the boundary, and are thus particularly relevant to describe a physical system with confined fermions.  

The pointwise vanishing of the normal spinorial current is achieved by a large family of boundary conditions. Within this family, we are interested in those satisfying an elliptic regularity property. Since the Dirac operator is a first-order elliptic operator, regular boundary conditions are those where the graph norm of the operator controls the $H^1$-norm in the domain.
A general way to guarantee this regularity is the Shapiro-Lopatinski condition, that translates the Fredholm property for general (pseudo)-differential operators with boundary conditions into an algebraic 
(or geometric) 
condition involving the principal symbol. 
We analyse this condition for Dirac operators with symmetric local boundary conditions.\medskip 

\subsection{Overview of previous results.}
In this paper, we consider a $d$-dimensional Riemannian manifold $M$, 
and a rank-$N$ complex vector bundle $\mathcal{S}\to M$ (more precisely a  bundle of Clifford modules, see below) whose sections
will be called spinors.
The main example is a
smooth domain $M \subset \R^d$ in Euclidean space with a trivial bundle,
but our discussion will be general.
Before describing the setting in detail, we quickly review some of the existing results on this subject in the mathematical physics literature.

The case of domains $M \subset \R^2$ has attracted considerable attention, since electronic excitations in graphene are described by a Dirac operator. For the actual description of graphene, $N=4$ spinor components are needed, but the case $N=2$ is a useful toy model. In this context, a local boundary condition called infinite-mass boundary condition has been introduced in the physics literature by \cite{BerryMondragon}. 
The regularity of these boundary conditions has been established in \cite{BFSV2017-s}.
It turns out that, in this case, at each point on the boundary there is a one-parameter space of possible boundary conditions that are regular except for two values corresponding to so-called \emph{zigzag} boundary conditions. 

For the actual graphene model with $N=4$, so-called \emph{armchair} or \emph{zigzag} boundary conditions arise from lattice terminations. Armchair conditions are regular (due to the equivalence with infinite mass boundary conditions shown in \cite{BFSV2017-g}), while zigzag boundary conditions are not, see e.g. \cite{holzmann2021zigzag}. 

In \cite{AkhmerovBeenakker}, the authors establish the general form of symmetric boundary condition for four-component spinors, i.e., the operator modeling graphene. In \cite{benguria2022block}, it is shown that these boundary conditions are generically regular, since the operator can be transformed in a direct sum of two-component operators, up to bounded corrections.

We are not aware of a comprehensive analysis of possible boundary conditions in dimension $3$ and beyond. 
For the case $d=3$ and $N=4$, the most-studied local boundary is the so-called \textsc{mit}-bag condition. It was introduced in the seventies by physicists at MIT to model hadrons as fermions confined in a \emph{bag}. This model was finally shown to be self-adjoint in \cite{BB13} or \cite{ourmieres2018strategy}.
A one-parameter family of boundary conditions called generalized \textsc{mit}-bag has been studied in \cite{behrndt2020self, rabinovich2021-BVP, arrizabalaga2023eigenvalue}, as a limiting case of delta-shell interactions.

These boundary conditions can be defined whenever a chirality operator exists, see Examples~\ref{ex:generalized_infinite_mass} and Corollary~\ref{cor:gen_MIT} below.

\medskip
In the Riemannian geometry setting, we refer to \cite{bookBooss}, \cite{BB13} and \cite{Ginoux} for an introduction to boundary value problems for Dirac operators. 
In some sense, the natural boundary condition from the geometry viewpoint is the nonlocal Atiya-Patodi-Singer (\textsc{aps}) boundary condition. It guarantees that spinors can be extended across the boundary, which in turn ensures the regularity of the corresponding Dirac operator, and allows to generalize index theorems for compact manifolds to the case of manifolds with boundary. Local boundary conditions have also been considered, in particular those coming from a chirality operator and the \textsc{mit}-bag boundary condition. (Note that, in the geometry, literature, the latter boundary condition is taken to be skew-symmetric.)

\subsection{Setting and notation.}
Before stating our results, we need some definitions to describe spinors on general manifolds. The chief example to keep in mind is a domain $M \subset \R^3$ with the trivial spinor bundle $\calS\define M\times \C^4$. For the convenience of the reader mainly interested in this case we first specify all the involved notions explicitly; the setting for the general case will be given below.\medskip 

So for now let $M$ be a domain in $\R^3$ with smooth boundary, 
and  $\mathcal{S}\define M\times \C^4$: A spinor is a section of this bundle, which in this case is just a function from $M$ to $\C^4$. The space of
smooth sections of a bundle $E$ will be denoted by $\calC(E)$, 
so that $\calC(\mathcal{S})$ is the space of spinors on $M$ and $\calC(T^*M)$ is the space of co-tangent vector fields (one-forms), which in the Euclidean case can be identified with smooth functions from $M$ to $\R^3$.
One then uses the mutually anticommuting Hermitian Dirac matrices $\gamma_1, \gamma_2, \gamma_3$ to define the operation of Clifford multiplication as a map from $\calC(T^* M) \times \calC(S)$ into $\calC(\mathcal{S})$ by the formula
$$
 (c_t \psi)(x) = \sum_{j=1}^3 t(x)_j \gamma_j \psi(x), \text{ for all } t\in \calC (M \times \R^3), \psi \in \calC(\mathcal{S}), \text{ and } x \in M.
$$
We define an inner product on spinors by
$$
\left( \psi, \phi \right) \define \int_M \langle \psi(x), \phi(x) \rangle_{\C^4} \di x,
$$
and denote the completion of $\calC(\calS)$ with respect to the corresponding norm by $L^2(\calS)$.
The Dirac operator is then defined as an unbounded operator in $L^2 (\calS)$ that acts as the differential expression
$$
D \psi \define - \ii \sum_{j=1}^3 \gamma_j \partial_{e_j} \psi =  - \ii \sum_{j=1}^3 c_{e_j}\partial_{e_j} \psi .
$$
In the geometry literature, it is usual to define Clifford-multiplication to be skew-adjoint (which amounts to multiplying the $\gamma$-matrices by $\ii$), and there is no $-\ii$ in the definition of the Dirac operator. 
Throughout this paper, we will consider the Clifford multiplication to be Hermitian. 
\medskip

In the general case, let $(M, g)$ be a compact Riemannian manifold with boundary, and let $\pi\colon \mathcal{S}\to M$ be a  bundle of Clifford modules over $M$. 
This means that there is a natural bundle map $\mathcal{C}\ell(T^*M)\otimes \mathcal{S}\to \mathcal{S}$
that induces, for each $p\in M$, an action of the Clifford algebra of $T^*_pM$ on the fiber $\mathcal{S}_p=\pi^{-1}(p)$. 
 The induced action of $v\in T_pM$ (identified with $T^*_pM$ using the metric $g$) on $\phi\in \mathcal{S}_p$ is denoted by $c_v\phi$, and satisfies
 $c_vc_w\phi+c_wc_v\phi=2g_p(v,w)\phi$  for $v,w\in T_pM$ and $\phi\in \mathcal{S}_p$. Moreover, 
we assume that $\mathcal{S}$ is equipped with a Hermitian bundle metric $\langle.,.\rangle$ (being antilinear in the second component) such that the Clifford action is self-adjoint,  $\mathcal{S}$ is also equipped with a metric connection $\nabla$ 
that satisfies the compatibility relation 
\begin{align}\label{eq:nabla_Cliff}
\forall \phi\in C^\infty(\mathcal{S})\qquad
\nabla_X (c_Y\psi)= c_{\nabla_X Y}\psi + c_Y \nabla_X \psi
\end{align}
for all smooth vector fields $X,Y$ on $M$. Such a bundle $\mathcal{S}$ is called a \emph{Clifford bundle} and the sections of $\mathcal{S}$ will be called \emph{spinors}  (cf. \cite[Def.~3.4]{Roe}).  
 
 The previous data give rise to
 a Dirac operator $D$ on sections of $\mathcal{S}$, given locally by
 \[
 D\psi = -\ii\sum_{j=1}^d c_{e_j}\nabla_{e_j}\psi,
 \]
 where $\{ e_j\}$ is any local orthonormal moving frame of $\mathcal{S}$.

Recall that we denote the (complex) rank of $\mathcal{S}$ by $N$. In order for the Clifford multiplication to exist it is necessary to have $N \ge 2^{\lfloor d/2 \rfloor}$. In Riemannian geometry, \emph{the} Dirac operator  (also called  \emph{Atiyah-Singer-Dirac operator}) is usually defined for $\mathcal{S}$ being the spinor bundle of a Riemannian spin manifold for which has $N= 2^{\lfloor d/2 \rfloor}$, while in physics the Riemannian manifold is often thought of as a space-like
slice of a $(d+1)$-dimensional Lorentzian manifold, and correspondingly $N= 2^{\lfloor (d+1)/2 \rfloor}$.

The Riemannian metric $g$ induces an $L^2$-inner product on spinors over $M$  by 
\[
( \psi, \varphi ) \define \int_M \langle \psi, \varphi\rangle dv,
\] 
where $dv$ is the volume measure associated to $g$.
As an operator in the first Sobolev space $W^{1,2} (\mathcal{S})$, $D$ satisfies
\begin{equation}
    \label{eq:Green}
    (D\psi, \varphi) - (\psi, D\varphi) = -\ii \int_{\partial M} 
    \langle c_\nu \psi , \varphi\rangle  d\sigma, 
\end{equation}
where $\nu$ is the outward unit normal, $d\sigma$ is the volume element of $\partial M$, 
and we abused the notation by writing $\psi$, $\varphi$ 
on the right-hand-side for their boundary traces.
It is clear from this expression that in order to obtain a symmetric operator, a boundary condition is needed.

\medskip 
Local boundary conditions on a spinor $\psi$ are imposed by requiring that,
at each $s\in\partial M$, the value $\psi(s)$ belongs in some 
subspace $\Lambda_s \subset \mathcal{S}_s$ of the fiber of $\mathcal{S}$ over $s$.  This differs from the definition of local boundary condition in \cite[Def.~7.9]{BB12} where \emph{regularity} as in Definition~\ref{def:regular} is assumed. But this definition is equivalent to the one used in \cite[Def. 6.25]{baer_bandara}, as follows from Lemma~\ref{lem_bound}.  More precisely we define:

\begin{definition} \label{def:local_BC}  Let $\mathcal{S}|_{\partial M}$ be the restriction of $\mathcal{S}$ to $\partial M$.
	A \emph{local smooth boundary condition} is a smooth 
	subbundle $\Lambda\subset \mathcal{S}|_{\partial M}$.\medskip 
	 
Given such a $\Lambda$,	the domain of the \emph{associated Dirac operator $D_{\Lambda}$} is
	\begin{equation*}
		\dom D_{\Lambda} = \overline{\{\psi \in \calC(\mathcal{S})\ |\ \psi(s) \in \Lambda_s,  \, \forall s \in \partial M \}}^{\Vert . \Vert_D},
	\end{equation*} 
	where $\overline{\phantom{A}.}^{\Vert . \Vert_D}$ denotes the closure in the graph norm $\Vert \psi\Vert_{D}^2\define \Vert \psi\Vert_{L^2}^2+\Vert D\psi\Vert_{L^2}^2$. Thus, by definition, $D_\Lambda$ is a closed operator.\medskip 
	
A boundary condition $\Lambda$ will be called \emph{symmetric} if and only if $D_\Lambda$ is symmetric, i.e. for all $\psi, \phi\in \dom\, D_\Lambda$ it holds that
	  \begin{align}\label{eq:symm} (D_\Lambda \psi, \phi) = (\psi, D_\Lambda \phi).\end{align}	  
\end{definition}

If one is just interested in self-adjointness, the smoothness requirement of the boundary condition can probably be relaxed to $C^\alpha$-regularity for some sufficiently large $\alpha$, but not completely omitted: In \cite{cassano2020self}, it is shown in a two-dimensional example 
 how jumps in the parameters defining the boundary conditions are an obstacle to essential self-adjointness (the case of a half plane is included in their results for sectors, by taking $\omega= \pi/2$).

\subsection{Main Results}

Our first result identifies the local smooth boundary conditions $\Lambda$ for which $D_\Lambda$ is self-adjoint.\medskip 

To formulate this result, let
 $E_\pm (c_\nu)$ denote the subbundles of $\mathcal{S}|_{\partial M}$ whose
fibers $E_\pm(c_{\nu(s)})$ are the $\pm 1$ eigenspaces of $c_{\nu(s)}\colon \mathcal{S}_s\to \mathcal{S}_s$.

\begin{theorem} \label{thm:local}
Let $(M, g)$ be a compact Riemannian manifold with boundary and let $\mathcal{S}$ be a Clifford bundle over $M$ of
rank $N= 2 n$.  Let $\Lambda\subset \mathcal{S}|_{\partial M}$ be a smooth subbundle. Then the following are equivalent:
\begin{enumerate}[(i)]
	\item $D_\Lambda$ is self-adjoint
 \item $D_\Lambda$ is symmetric and $\mathrm{rank}\, \Lambda=n$
\item $\Lambda$ is the graph of fiber-wise unitary bundle map 
from $E_+(c_\nu)$ to $E_-(c_\nu)$, i.e., there exists a smooth map $F: E_+(c_\nu) \to E_-(c_\nu) $ such that, for each $s$, $F_s$ is unitary and $\Lambda_s = \{\phi + F_s \phi |\ \phi \in E_+(c_{\nu(s)})\}$.
\end{enumerate}
\end{theorem}

We want to emphasize that the above results mean that for local smooth boundary conditions the question of self-adjointness translates into pointwise linear algebraic conditions. This is an important advantage of our choice to define $D_\Lambda$ as the closure in the graph norm, rather then restricting it to $H^1$ as in e.g. \cite{BFSV2017-s, ourmieres2018strategy}.\medskip 

The equivalence of (i) and (ii) follows directly from Lemma~\ref{lem:adjoint_BC} (as shown below that lemma), 
and is an immediate application of the general theory on boundary values for Dirac operators 
to the special case of local smooth boundary conditions. The equivalence of (ii) and (iii) will be proven in Section~\ref{sec:symmetricBC}. The proof amounts to a description of the Grassmannian of $n$-dimensional subspaces $\Lambda_s\subset (\mathcal{S}|_{\partial M})_s$ for which the symmetry condition $\langle c_{\nu(s)} u , v \rangle =0$ holds for all $u, v\in\Lambda_s$.
This Grassmannian was determined by Arnold in \cite{arnold2000complex}. (We also point out that in the recent preprint \cite{jud2024classifying}, a Dirac operator with a second order regularization is studied, and the authors are 
also led to a study of a Grassmannian.) \medskip 

Since the spaces $E_\pm(c_{\nu(s)})$ are both isomorphic to $\C^n$, for each $s$
there is an $n^2$-dimensional space of unitary transformations between them. However, it may not be possible to define such a map continuously globally on $\partial M$.
In Section~\ref{subsec:global}, we study the existence of self-adjoint local smooth boundary conditions. The following theorem shows that there might be obstructions to the existence of a self-adjoint local smooth boundary condition. (Recall that $d=\mathrm{dim}\, M$ and $N= \rank\, \calS$.)
\begin{theorem}
    In the cases $d=3$, $N=2$ and $d=5$, $N=4$, a self-adjoint local smooth boundary condition exists if and only if there exists a unit tangent vector field on $\partial M$.
\end{theorem}

A more detailed statement is given in Proposition~\ref{prop:global_d_odd} below.
This theorem implies that for bounded simply connected domains $M \subset \R^3$, it is not possible to give a local symmetric boundary condition for $2$-component spinors (Weyl spinors). A general analysis of the existence of self-adjoint boundary condition can probably done with the help of algebraic topology. We restrict here to low-dimensional examples and the cases where a chirality operator exists.

 \begin{definition}
     A \emph{chirality operator}\label{def:chiral} $\beta\colon \mathcal{S}\to \mathcal{S}$  is a unitary endomorphism of vector bundles that commutes with the connection, anticommutes with Clifford multiplication, and satisfies $\beta^2=\Id$.
 \end{definition}

Such an operator  splits the spinor bundle into a direct sum $\mathcal{S} = \mathcal{S}^{+}\oplus \mathcal{S}^{-}$, where the fibers of $\mathcal{S}^\pm$ are the $\pm 1$ eigenspaces of $\beta$. For $d$ even and $M$ orientable,  Clifford multiplication with the complex volume form gives such an operator $\beta$, \cite[p.~32]{Ginoux}. 

In the Euclidean setting the chirality operator exists in odd dimension $d$ and rank $N=2^{\lfloor{(d+1)/2}\rfloor}$, and acts by the Dirac matrix to the additional unit vector in the euclidean space of one dimension higher. Physically, considering massive spinors requires the existence of a chirality operator.

\begin{theorem}
    Let $(M,g)$ be a compact Riemannian manifold with boundary and let $\mathcal{S}$ be a Clifford bundle over $M$ of
rank $N= 2 n$, with a chirality operator whose eigenbundles will be called $\mathcal{S}^\pm$. Then, 
self-adjoint local smooth boundary conditions are in one-to-one correspondence with smooth sections of the bundle 
$U(\mathcal{S}^+|_{\partial M})\to \partial M$ of fiber-wise unitary endomorphisms of $\mathcal{S}^+|_{\partial M}$.
\end{theorem}

This correspondence is proved in Proposition~\ref{prop:BC_split}, where we also show the explicit form of the boundary condition generated by a section of $U(\mathcal{S}^+|_{\partial M})$. In particular, if a chirality operator exists, there are always infinitely-many
local smooth symmetric boundary conditions.

\medskip
In addition to self-adjointness we want to describe local smooth boundary conditions that are elliptic. As already
mentioned,
the key tool to prove elliptic regularity for local boundary conditions is the standard Shapiro--Lopatinski criterion (see e.g. \cite[Chapter \textsc{xx}]{HormanderIII} for general pseudo-differential operators or \cite{bookBooss} for the Dirac case). 
In Section~\ref{sec:SPcriterion} we investigate this criterion for local smooth boundary conditions more closely.
In particular, we show that the Shapiro-Lopatinski condition is \emph{sharp} in our setting, 
in the sense that if it fails at a point on the boundary, 
the corresponding boundary condition is not regular (see Proposition~\ref{prop_revSL}).
\medskip

In Sections 5 and 6 we apply the previous results to investigate 
the existence of elliptic symmetric boundary conditions in 
various settings.  
In Section 5 we consider the case when a chirality operator is 
present.  We obtain a general description 
of such conditions, see Propositions \ref{prop:BC_split} and
\ref{prop:LS-for-split}, and in \S 6 we study these 
conditions in low dimensions.
A consequence of our analysis is that, in dimensions $3$ and $5$, no regular symmetric local boundary conditions exist for \emph{the} Dirac operator from Riemannian geometry, i.e., with $N = 2^{\lfloor d/2 \rfloor}$. Again, regular local smooth symmetric boundary conditions can be found by doubling the dimension of the spinor bundle. For the low-dimensional cases $d=3,4$ and $N=4$, we find an explicit parametrization of the boundary conditions satisfying the Shapiro-Lopatinski condition, see
Propositions~\ref{prop:SL_in_d=3-general} and~\ref{prop:SL_in_d=4-new}.\medskip 

In Section~\ref{sec:transm} we consider transmission conditions. Our definition includes the transmission boundary conditions that appear in the setting of Dirac operators with singular potentials in mathematical physics, see Example~\ref{ex_transmissII}. Although transmission conditions are not local boundary conditions in the sense of Definition~\ref{def:local_BC}, our results still apply after an appropriate transformation of the problem, see Theorem~\ref{thm:boundary_transmission}. 
\medskip

Finally, we note that our theorems remain valid if one adds to the Dirac operator a zero order differential operator that  is bounded 
and symmetric  in $L^2(\calS)$ (e.g. magnetic, electrostatic  and/or mass-type $L^\infty$-potentials).\medskip

\noindent \textbf{Acknowledgements.}
Our collaboration leading to this paper originates from the hybrid workshop \emph{Analytic and Geometric Aspects of Spectral Theory} at Casa Matemática Oaxaca in August 2022. We are grateful to the organizers and BIRS for hosting this stimulating event.
H. VDB. received financial support from  the Center for Mathematical Modeling (Universidad de Chile \& CNRS IRL 2807) through ANID/Basal project \#FB210005 and from ANID/Fondecyt project \#1122--0194. 
We also thank Konstantin Pankrashkin for suggesting to examine transmission boundary conditions as well.

\section{Preliminaries on boundary values for Dirac operators} \label{sec:Prelim}

In this section we give a short overview of the  general theory of boundary conditions from \cite{bookBooss, HormanderIII, BB13}. We restrict to the case where $M$ is a compact manifold with boundary (for non-compact $M$ with compact boundary, see Remark~\ref{rem_noncomp}). 
We need to express $D$ in a special form in a tubular neighborhood of the boundary.  Such a neighborhood is diffeomorphic to $\partial M\times [0,\epsilon)$, where $t \in [0, \epsilon)$ runs along unit normal geodesics to the boundary. When using the pull-back of the embedding $\iota\colon \partial M\times [0,\epsilon) \to M$, we can and will in the following consider sections of $\mathcal{S}|_{\iota(\partial M\times [0,\epsilon))}$ as $t$-dependent sections over $\partial M$. Then, the Dirac operator can be written as  
\begin{equation}\label{eq:def_At}
	D= -\ii c_\nu \left(\partial_t + A_t\right),
\end{equation} 
where $A_t$ are first order differential operators acting on $C^\infty (\mathcal{S}|_{\partial M})$. For $t=0$, $A_0$ is a Dirac-type operator on the closed manifold $\partial M$ over the Clifford bundle $\mathcal{S}|_{\partial M}$, and therefore self-adjoint with domain $H^1(\mathcal{S}|_{\partial M})$. We will denote by $a(s,k)\in \End (\mathcal{S}_s) $ the principal symbol of $A_0$ (where $s\in \partial M$, $k\in T_s^*\partial M$). It is defined by 
$$
(A_0 f u)(s) = (f A_0 u)(s) + a(s, \romd f|_s) u  (s)  \quad \forall u \in C^\infty (\mathcal{S}|_{\partial M}),   f \in C^\infty (\partial M).
$$
Since $\partial M$ is compact, the spectrum of $A_0$, denoted by $\sigma (A_0)$, is discrete. For any $I \subset \R$, we can define the generalized Sobolev spaces on the boundary
$$
H^s(P_I A_0) \define\left\{u  \in L^2(\mathcal{S}|_{\partial M}) \left| \sum_{\lambda_j \in \sigma(A_0) \cap I} |\lambda_j|^{2s} \norm{P_{\lambda_j} u}^2 < \infty \right.\right\}
$$
where $P_{\lambda_j}$ is the orthogonal projection onto the eigenspace associated to $\lambda_j$. The $P_I$ in the notation of the space indicates that it is the generalized  Sobolev space associated to the spectral projection of $A_0$ corresponding to the interval $I$.

Since $A_0$ is a first-order elliptic operator on the closed manifold $\partial M$, we have $H^s(A_0)\define H^s(P_{\R}A_0) \cong W^{s,2}(\mathcal{S}|_{\partial M})$, the standard $L^2$-based Sobolev space of order $s$ on $\calS|_{\partial M}$.

As usual, the boundary trace $t_{\partial M}\colon \calC (\mathcal{S})\to \calC (\mathcal{S}|_{\partial M})$ extends to a bounded operator from $W^{s,2}(\mathcal{S})$ to $W^{s-1/2,2}(\calS|_{\partial M})$ for $s>\frac{1}{2}$. By \cite[Thm.~3.2]{BB13} the trace map can be extended further to 
$$\dom\, D_{\max} \define \{ u \in L^2(\mathcal{S})\ |\  Du \in L^2(\mathcal{S})\}
$$
--the domain of the \textit{maximal Dirac operator}--and for any $\lambda \notin \sigma(A_0)$ we have
\begin{align}\label{eq:trace_max}
t_{\partial M}(\dom\, D_{\max})  = \check H (A_0) \define H^{1/2}( P_{(-\infty,\lambda)} A_0) \oplus  H^{-1/2}( P_{(\lambda,+\infty)} A_0).
\end{align}
If $\dom\, D_{\max}$ is equipped with the graph norm and $\check H (A_0)$ with the norm induced by the $H^{\pm 1/2}$-norms, this extension is continuous. \medskip 

By \cite[p.~4+6]{BB12}, the adjoint of the maximal Dirac operator $D_{\mathrm{max}}$ is the minimal Dirac operator $D_{\mathrm{min}}$ whose domain equals
$$\dom\, D_{\mathrm{min}} = \overline{ \calC_{cc} (\mathcal{S})}^{\Vert . \Vert_D}=\{ \phi \in \dom\, D_{\max}\ |\ t_{\partial M}\phi =0\}, 
$$
where $\calC_{cc} (\mathcal{S})$ are the smooth spinors compactly supported in the interior of $M$. For nonempty boundaries $\dom\, D_{\min}$ is strictly contained in $\dom\, D_{\max}$.\medskip  

For future reference we collect some more results from \cite{BB12}:

\begin{lemma}\label{lem_prel_D}\hfill 
	\begin{enumerate}[(i)]
\item \cite[Prop.~7.2]{BB12} The closed extensions $D\colon \dom\,  D\subset L^2(\calS)\to L^2(\calS)$ of $D_{\mathrm{min}}$ are in one-to-one correspondence with closed subsets $B\subset \check{H}(A_0)$ via  $B=t_{\partial M}(\dom\, D)$. 
\item \cite[Lem.~6.1]{BB12} There is a partial inverse of the trace operator--the \emph{extension operator} 
\[ \mathcal{E}\colon \check H (A_0)\to \dom\, D_{\max},\] with $t_{\partial M}\mathcal{E} = \Id_{\check H(A_0)}$ and $\mathcal{E}(\calC (\calS|_{\partial M}))\subset \calC(\calS)$.
\item \cite[Lem.~6.3]{BB12} The functional
$\check H (A_0)\times \check H (A_0)\to \mathbb C$, $(\phi, \psi) \mapsto \int_{\partial M} \langle c_{\nu}\phi, \psi \rangle$ is continuous.
 \item \cite[Thm. 6.7]{BB12} The 
Green's identity \eqref{eq:Green} is valid for all  $\phi, \psi \in \dom (D_{\max})$.
\item \cite[Sec.~7.2]{BB12} Let $D\colon \dom\,  D\subset L^2(\calS)\to L^2(\calS)$ be a closed extension of $D_{\mathrm{min}}$ with $B=t_{\partial M}(\dom\, D)$. Then, its adjoint operator is the closed extension corresponding to \[B^{\mathrm{ad}}\define \left\{\psi\in \check{H}(A_0)\ \Bigg|\  \int_{\partial M} \langle c_{\nu}\phi, \psi \rangle=0\ \forall \phi\in B\right\}.\]
\item \cite[Thm. 6.7]{BB12}  $\calC(S)$ is dense in $\dom\ D_{\max}$  
with respect to the graph norm.
\end{enumerate} 
\end{lemma} 

The next two Lemmas specify these results to local boundary conditions.
We first identify the closed subspace corresponding to such a boundary condition.
\begin{lemma}\label{lem_bound} For a local smooth boundary condition $\Lambda\subset \calS|_{\partial M}$ we have \[B_\Lambda \define t_{\partial M} (\dom\, D_\Lambda) = \overline{\calC(\Lambda)}^{\Vert . \Vert_{\check H}}.\]
\end{lemma} 

\begin{proof} 
	The inclusion $\supset$ directly follows from the definition of $D_\Lambda$ and the fact that $B_\Lambda$ is defined to be closed in $\check H (A_0)$. For the other inclusion, let $\psi\in B_\Lambda$. Then $\phi \define \mathcal{E}\psi \in \dom\, D_\Lambda$ by Lemma~\ref{lem_prel_D}(i). Hence, by the Definition~\ref{def:local_BC} of $\dom\, D_\Lambda$ there is a sequence $\phi_n\in \calC(S)$ with $\phi_n|_{\partial M}\in \calC(\Lambda)$ converging to $\phi$ in the graph norm. The continuity of $t_{\partial M}\colon \mathrm{dom}\, D_{\max}\to \check H(A_0)$ thus implies that $\phi_n|_{\partial M}\to t_{\partial M}(\phi) =\psi$ in $\check H(A_0)$ which proves the claim.
\end{proof}

Next, we identify the boundary space for the adjoint operator.

\begin{lemma} \label{lem:adjoint_BC}
 If $\Lambda$ is a local smooth boundary condition, then its adjoint $(D_\Lambda)^*$ is the Dirac operator $D_{\Lambda^*}$ with local boundary condition
 $\Lambda^*$ defined by 
  \begin{equation} \label{eq:adjoint_BC}
     	v \in \Lambda_s^*  \quad \text{ if and only if for all }  u\in\Lambda_s, \quad \langle  c_{\nu(s)} u, v\rangle  = 0.
 \end{equation}
 In particular, 
$\Lambda^*$ is a smooth subbundle of $\calS|_{\partial M}$ and $\rank\, \Lambda +\rank\, \Lambda^* = N$.
\end{lemma}

In view of this, $D_\Lambda$ is a symmetric operator if $\Lambda \subset \Lambda^*$, which amounts to
\begin{equation} \label{eq:sym_BC}
    \quad \langle  c_{\nu(s)} u, v\rangle  = 0 \quad  \text{ for all } u, v \in \Lambda_s, s\in \partial M, 
\end{equation}
and a self-adjoint operator if it is symmetric and $\rank\, \Lambda = N/2$. Throughout the paper, we will call such boundary conditions symmetric (resp. self-adjoint).

\begin{proof}
    By the definition of an adjoint, $(D_\Lambda)^*$ is a closed operator with domain a subset of $\dom D_{\max}$. By Lemma~\ref{lem_prel_D}.(i) it is characterized by the boundary traces $t_{\partial M}(\dom (D_{\Lambda})^*)$. 
    The inclusion $\dom\, D_{\Lambda^*} \subset \dom\, (D_{\Lambda})^*$ follows from the Green's identity \eqref{eq:Green}, which holds in the maximal domain in view of Lemma~\ref{lem_prel_D}.

    For the opposite inclusion, assume that $\psi \in \calC(\calS) \cap \dom\, (D_{\Lambda})^*$. Then, for all $\phi \in t_{\partial M}(\dom\, D_\Lambda)$, we have by the definition of an adjoint and Green's identity 
    $$
\int_{\partial M} \langle c_\nu \psi|_{\partial M}, \phi \rangle = 0 .
    $$
    Fix $s \in \partial M$ and $v \in \Lambda_s$ and take a sequence $\phi_n \in \calC( \Lambda)$ that converges, in the sense of distributions on $\partial M$, to $\delta_s v $. Such a sequence can be constructed explicitly by considering a trivialization of $\calS|_{\partial M}$ in a neighborhood of $s$ and suitable $n$-dependent scalar functions. 
    Then, by using the extension operator, $\calE \phi_n \in \dom\, D_\Lambda$ and hence, 
    $$
    \int_M \langle c_\nu \psi|_{\partial M}, \phi_n \rangle =  0 \quad \text{ for all } n \in \N.
    $$
    Taking the limit as $n \to \infty$ shows that
    $$
\langle  c_{\nu(s)} \psi(s), v\rangle = 0
    $$
    and hence $\psi|_{\partial M}\in \calC(\Lambda^*)$.
  Since smooth spinors are dense (for the graph norm) in $\dom\, (D_\Lambda)^*$, this shows the inclusion $\dom\, (D_\Lambda)^*  \subset \dom \, D_{\Lambda^*}$. 

  Finally, if $\rank \Lambda= k$, 
  $$\rank\, \Lambda^* = \dim\, \Lambda^*_s= N- \dim (\Ker c_\nu P_{\Lambda_s}) = N-k$$ 
  since $c_{\nu(s)}$ is nondegenerate.
\end{proof}

\begin{remark}\label{rem_noncomp} For all of the above, the assumption that the manifold $M$ is compact can be  relaxed to $M$ being complete and having a compact boundary, as done in \cite{BB13, BB12}. For Lemma~\ref{lem_prel_D}.vi the additional assumption that $D$ is complete is needed, cf. \cite[Thm. 6.7]{BB12}. For self-adjointness and regularity in case $M$ is complete with compact boundary, see Remarks~\ref{rem_noncomp_sa} and~\ref{rem_noncomp_reg} at the end of the corresponding sections.
\end{remark}

\section{Self-adjoint boundary conditions.} \label{sec:symmetricBC}

By the last section, we know that for local smooth boundary conditions being self-adjoint (in the sense that the associated Dirac operator is self-adjoint) is equivalent to being symmetric with maximal rank. The goal of this section is to prove the equivalence of the characterization (ii) and (iii) of Theorem~\ref{thm:local}. 

\subsection{Pointwise considerations} \label{subsec:pointwise}
Let $s\in \partial M$. The symmetry condition (\ref{eq:adjoint_BC}) suggests to consider the sesquilinear form:
\begin{equation}\label{}
	\frakb_s\colon \calS_s\times \calS_s\to\bbC,\qquad
	\frakb_s(u,v) \define \inner{c_{\nu(s)}u}{v}.
\end{equation}
The condition for being symmetric is then that
$\forall u, v\in\Lambda_s$ $\frakb_s(u,v)= 0$ at each point $s\in\partial M$. 
That is, $\forall s\in\partial M$ the subspace
$\Lambda_s\subset \calS_s$ should be a complex subspace that is \emph{isotropic} with respect to $\frakb_s$. Since we are interested in self-adjoint local smooth boundary counditions, we need $\Lambda_s$ to be $n$-dimensional 
 (Recall that the complex
 dimension of $\calS_s$ is $N=2n$.) by Lemma~\ref{lem:adjoint_BC}.
 %
%
We first analyse, for a fixed $s$, the space of all such $\Lambda_s$, following
\cite{arnold2000complex}.  Accordingly, we will
drop the subscript $s$ from the notation in $\calS_s$, $\nu(s)$ and so on.\medskip 

Recall that we take Clifford multiplication to be Hermitian. Since we took $\nu$ a normalized outward normal, $c_\nu^2 = \Id$. Also, $c_\nu$ anticommutes with $c_t$ for all $t\in T_s\partial M$. Therefore, the operator $c_\nu$ has  eigenvalues $\pm 1$ with equal multiplicity $n$, which implies that
$\frakb$ is non-degenerate.  It follows that all isotropic
subspaces of $\calS$ have dimension at most $n$, and 
therefore we are interested in all subspaces 
$\Lambda\subset \calS$ that are maximally isotropic with
respect to $\frakb$.

Since 
\begin{equation}\label{eq:hermConj}
	\frakb(v,u) = \overline{\frakb(u,v)},
\end{equation}
the imaginary part of $\frakb$
\begin{equation*}
	\Omega\define \Im \frakb
\end{equation*}
is a (real) symplectic form.  Note that it satisfies
\[
\forall \lambda\in S^1\qquad \Omega(\lambda u, \lambda v) = \Omega(u,v),
\]
i.e., it is \emph{$\bbC$-symplectoidal} in the sense of Arnold.

\medskip

\begin{lemma}\label{lem:Lagr_b_Omega}
A complex $n$-dimensional subspace $\Lambda\subset \calS$ is 
maximally isotropic with respect 
to $\frakb$ if and only if it is Lagrangian (=maximally isotropic)
with respect to $\Omega$.
\end{lemma}
\begin{proof}
	It is obvious that maximally isotropic with respect
 to $\frakb$ implies Lagrangian, since the dimension is $n=N/2$. For the converse, assume that $\Lambda$ is $\Omega$--Lagrangian.  
	Then simply note that
	\[
	\forall u,\,v\in\Lambda\qquad \frakb(\ii u,v) = \ii\frakb(u,v)
	\]
	is both real and purely imaginary  (by the Lagrangian condition together with $\ii u\in\Lambda$), 
	so it must be zero.  Thus $\Lambda$ is isotropic with respect to $\frakb$.
\end{proof}

\begin{definition}
Denote by $\calL$ the space of all complex $n$-dimensional subspaces $\Lambda\subset \calS$ that are isotropic with respect 
to $\frakb$.  In view of the last lemma such a $\Lambda$ will be called an \emph{$\Omega$-Lagrangian subspace}.
\end{definition}

We now review the main result of Arnold, \cite{arnold2000complex}, that $\calL$ is diffeomorphic to the unitary group $U(n)$
(Corollary \ref{cor:conclusion} below).
Let
\begin{equation*}\label{}
	E_{\pm} \define \pm 1\text{ eigenspace of } c_\nu,\qquad 
 \calS = E_+\oplus E_-.
\end{equation*}

\begin{lemma}
Any $\Omega$-Lagrangian in $\calS$ is transverse to each $E_{\pm}$, and therefore is the graph of a (unique) 
map $F\colon E_+\to E_-$.
\end{lemma}
\begin{proof} Let $\Lambda$ be a $\Omega$-Lagrangian in $\calS$.  Since $E_\pm$ and $\Lambda$ are both $n=N/2$ dimensional vector spaces, $\Lambda$ is transverse to $E_\pm$ if and only if $\Lambda\cap E_\pm =\{0\}$.  
Note that the restriction of $\frakb$ to $E_{+}\times E_{+}$ (resp. $E_- \times E_-$) coincides with the restriction of 
$ \inner{\cdot}{\cdot}$ (resp. $ - \inner{\cdot}{\cdot}$ ) to this space.
Let $v\in V\define \Lambda\cap E_+$.  Then $0=\frakb{(v,v)}=\inner{c_\nu v}{v}=\inner{v}{v}=\norm{v}^2$ and, thus, $V=\{0\}$. 	
	The proof that $\Lambda\cap E_- = \{0\}$ is identical. Since $\calS=E_+\oplus E_-$, the transversality implies that $\Lambda$ is the graph of a unique map $F\colon E_+\to E_-$. 
\end{proof}
For a  complex linear map $F\colon  E_+\to E_-$ let  $\Lambda_F$ be its graph, that is
\[
\Lambda_F = \left\{ u+F(u)\:|\: u\in E_+\right\} \subset E_+ \oplus E_- =\calS.
\]
Such an $F$ is called  \emph{unitary} if 
\[
\forall u,v\in E_+\qquad \inner{F(u)}{F(v)} = \inner{u}{v}.
\]

\begin{lemma}
Let $F\colon  E_+\to E_-$ be a $\C$-linear map.
Then its graph $\Lambda_F$ is an $\Omega$-Lagrangian subspace of $\calS$ if and only if $F$ is unitary.
\end{lemma}

\begin{proof}
	Fix a unitary map $F$. By Lemma~\ref{lem:Lagr_b_Omega} we need to investigate when $\Lambda_F$ is isotropic for $\frakb$.  Let $u,v\in E_+$, and compute
	\[
	\frakb(u+F(u), v+F(v)) = \inner{c_\nu(u+F(u))}{v+F(v)} = \inner{u-F(u)}{v+F(v)}.
	\]
	But
	\[
	\inner{u}{F(v)} = 0 = \inner{F(u)}{v}
	\]
	because the $E_\pm$ eigenspaces are orthogonal.  Therefore
	\begin{align}\label{eq:b_unit}
	\frakb(u+F(u), v+F(v))= \inner{u}{v} -\inner{F(u)}{F(v)},
	\end{align}
	and the conclusion follows since $F$ is unitary. 
	Is now $\Lambda$ an $\Omega$-Lagrangian in $\calS$, then the last lemma gives a unique complex linear map $F$ with $\Lambda=\Lambda_F$. So the only statement left to show is $F$ being unitary. But the same calculation as above leads to \eqref{eq:b_unit} where now the left-hand-side is zero by Lemma~\ref{lem:Lagr_b_Omega}.
\end{proof}

Altogether this implies

\begin{corollary}\label{cor:conclusion}
The space $\calL$ of all maximally isotropic subspaces of $(\calS, \frakb)$
is naturally isomorphic to the space of 
unitary maps $E_-\to E_+$.
\end{corollary}

\subsection{Global considerations} \label{subsec:global}
Recall that $d=\mathrm{dim}\, M>1$ and $N=2n=\mathrm{rank}\, \calS$. From the last section we directly obtain the proof of the remaining part of  Theorem~\ref{thm:local}:

\begin{proof}[Proof of  Theorem~\ref{thm:local}] The equivalence of (i) and (ii) was proven below Lemma~\ref{lem:adjoint_BC}.
	
Let $\Lambda$ be a local smooth boundary condition of rank $n$. Let $s\in \partial M$. By the last subsection $\Lambda_s$ is symmetric if and only if there is a unitary map $\widetilde{f}(s)\colon E_-(c_\nu(s))\to E_+(c_\nu(s))$. Since $\calS|_{\partial M}$ and the eigenspace bundles $E_\pm (c_\nu)$ are smooth, $\Lambda$ is a smooth subbundle of $\calS|_{\partial M}$ if and only if $\widetilde{f}$ is a smooth section of $U(E_+(c_\nu), E_-(c_\nu))$.
\end{proof}

We will see in Section~\ref{sec_chiral} that there will always be many self-adjoint local smooth boundary conditions if the Clifford bundle admits a chirality operator. Without the existence of a chirality operator on a given Clifford bundle, there can be topological obstructions to the existence of a self-adjoint local smooth boundary condition. We note that without the condition of being self-adjoint  there are of course always symmetric local boundary conditions, e.g. $\Lambda_s=\{0\}$.\medskip 

In case the boundary $\partial M$ admits a non-vanishing tangent vector field, there always exists a self-adjoint local smooth boundary condition, as can be seen in the next lemma. On the other hand, we will give in Proposition~\ref{prop:global_d_odd} a complete characterization of self-adjoint local smooth boundary conditions for $(d,N) =(3,2)$ and $(d,N)=(5,4)$, which are exactly the minimal $N$-values for the spinor bundle in dimension $3$ resp. $5$. In these cases, self-adjoint local smooth boundary conditions do only exist if there exists a non-vanishing tangent vector field on $\partial M$. For dimension $3$ this implies for example that the each boundary component is homeomorphic to a torus.\medskip 

\begin{lemma} \label{lem:tangent_Lambda}
	Let the unit tangent bundle of $\partial  M$ admit a smooth section $t$. Then the positive eigenbundle of $c_t$ is a self-adjoint local smooth boundary condition.
\end{lemma}

\begin{proof}
	Since $t$ is tangent in each point of $\partial M$, $c_t$ anticommutes with $c_\nu$.  Thus, $c_{t(s)}|_{E_+(c_{\nu(s)})}\colon E_+(c_{\nu(s)}) \to E_-(c_{\nu(s)})$. Since $c_t^2=\Id$ and since Clifford multiplication is self-adjoint, this map is unitary and hence, by Theorem~\ref{thm:local} $\Lambda_s=\{v+c_{t(s)}v \ |\ v\in E_+(c_{\nu(s)})\}$ defines a self-adjoint local smooth boundary condition. From $c_t^2=\Id$   we get  $c_{t(s)}\Lambda_s= \Lambda_s$. Thus, $\Lambda$ is the positive eigenbundle of $c_t$.
\end{proof}

\begin{remark}\label{rem_noncomp_sa}
For complete manifolds with compact boundary Theorem~\ref{thm:local} is still valid: This follows directly by gluing together the result for compact manifolds with boundary with the fact that for complete manifolds without boundary the Dirac operator is self-adjoint with domain $\dom\, D_{\mathrm{min}}=\dom\, D_{\mathrm{max}}$.
\end{remark}

\section{Regularity of self-adjoint local smooth boundary conditions}\label{sec:SPcriterion}

In this section  we study the regularity of self-adjoint local smooth boundary conditions. Here $M$ is still compact (except in Remark~\ref{rem_noncomp_reg}).

\begin{definition}\label{def:regular}
	A boundary condition is called \emph{regular} if $\dom\, D_\Lambda\subset W^{1,2}(\calS)$. 
    A boundary condition is \emph{strongly regular} or \emph{$\infty$-regular}, see  \cite[Def.~2.3]{BLZ} or \cite[Def.~7.15]{BB12}, if $\{\phi \in \dom\, D_\Lambda\ |\ D^\ell \phi\in L^2(\calS)\ \forall \ell \leq k\}\subset W^{k,2}(\calS)$ for all $k\in \mathbb N_{\geq 0}$.
\end{definition}

Strong regularity is equivalent to elliptic estimates, i.e. for all $k\in \mathbb N$ there is a $C_k>0$ such that
\[ \Vert \phi\Vert_{W^{k,2}(\calS)} \leq C_k \sum_{j=0}^k \Vert D^j\phi\Vert_{L^2(\calS)}\]
for all $\phi\in \dom\, D_\Lambda  \cap \bigcap_{j=1}^{k}\dom\, D^j$.
Analogously, regularity is equivalent to this estimate for $k=1$.

	\begin{remark} Any function in the domain $D_\Lambda$ is in $W^{1,2}_{\rm loc}$ away from the boundary. Thus, if $M$ is compact, the regularity is determined by the regularity of its boundary values. This can be made formal by using the trace and extension operators from Section~\ref{sec:Prelim}, which give the following characterization:
    \end{remark}
    
	\begin{proposition}
	    A local smooth boundary condition $\Lambda$ is regular if and only if $\overline{\calC (\Lambda)}^{\check H } \subset W^{1/2,2}(\calS|_{\partial M})$
	\end{proposition}

						\begin{proof} 
			If the boundary condition is regular, then taking the trace implies that $t_{\partial M} (\dom\, D_\Lambda)\subset W^{1/2,2}(\calS|_{\partial M})$ and hence $\overline{\calC (\Lambda)}^{\check H } \subset W^{1/2, 2}(\calS|_{\partial M})$ by Lemma~	\ref{lem_bound}. For the reciprocal, assume that $\overline{\calC (\Lambda)}^{\check H } \subset W^{1/2, 2}(\calS|_{\partial M})$. Then, using the extension operator from Lemma~\ref{lem_prel_D}(ii) and Definition~\ref{def:local_BC} of $\dom\, D_\Lambda$ we obtain $\dom\, D_\Lambda\subset W^{1,2}(\calS)$. 
			\end{proof}

The Shapiro-Lopatinski condition characterizes this inclusion of boundary traces in terms of the principal symbol of the boundary operator.

\begin{proposition} (Shapiro-Lopatinski criterion \cite[Remark~2.8]{BLZ} (see also \cite[Sec. 20.1]{HormanderIII}))\label{prop:shap-lop}
Let $\Lambda\subset \calS|_{\partial M}$ be a local smooth boundary condition. 
$D_\Lambda$ is strongly regular
if  for all $s \in \partial M $ and $k \in T_s^*\partial M$ with $|k|=1$,
$$
E_{+\ii}(a(s,k)) \cap \Lambda_s^\perp = \{0\}$$
or, equivalently, if  for all $s \in \partial M $ and $k \in T_s^*\partial M$ with $|k|=1$,
\begin{align}\label{eq:shap} 
E_{-\ii}(a(s,k)) \cap \Lambda_s = \{0\}\end{align}
where  $a$ is the principal symbol of $A_0$ and $E_{\pm \ii}(a(s,k))$ are the $\pm \ii$-eigenspaces of $a(s,k)\colon \calS_s\to \calS_s$. 
\end{proposition}

In our setting, the Shapiro-Lopatinski criterion is also a necessary condition for regularity. 
Since we were unable to locate a proof of this 
result in the literature, we provide it below. 

\begin{proposition}\label{prop_revSL}
    If there is $ s_0 \in \partial M$ such that the Shapiro-Lopatinski condition fails at $s_0$, in the sense that there is $\xi_0 \in T^*_{s_0}\partial M$ and $ v_0 \in \calS_{s_0}$ such that 
    $$
\|\xi_0\| = \|v_0\|= 1 \text{ and } v_0 \in  \Lambda_{s_0} \cap E_{-\ii}(a(s_0,\xi_0)),  
    $$
    then $D_\Lambda$ is not regular.
\end{proposition}

\begin{proof}
    We will show that on $\dom\, D_\Lambda$ the graph norm of $D_\Lambda$ is not equivalent to the $W^{1,2}$-norm. For that we will construct a sequence $\psi_n\in \dom\, D_\Lambda$ with bounded graph norm but unbounded $W^{1,2}$-norm. The support of $\psi_n$ will be a neighbourhood of $s_0\in U$ and will shrink to $s_0$ as $n\to \infty$. Throughout this proof, $C$ will denote a positive number independent of $n$, whose value can change from line to line.

For simplicity we first treat the Euclidean case, with a trivial bundle and a flat boundary in a neighborhood of $s_0$. We also assume that $\Lambda_s$ is independent of $s$ in this neighborhood. In the second part of the proof, we treat the general case. The key idea is that the flat case is a good approximation since the supports of the test functions shrink to the point $s_0$ on $\partial M$.

\noindent\textbf{Step 1 -- Euclidean case.} Let $M=\R_+^d=\{(s,t)\ |\ s\in \R^{d-1}, t\geq 0\}$. Let $\chi\colon \R_+ \to \R_+$ be a smooth decreasing cut-off function with support in $[0,1]$ and values in $[0,1]$. For each $n \in \N$ define a radius $r_n=n^{-1/2}$ and the cut-off function $\chi_n(s)= r_n^{-(d-1)/2}\chi( |s-s_0|/ r_n  ) $. The normalization is chosen such that 
\begin{equation} \label{eq:est_chi_n}
    \| \chi_n \|_{L^2(\R^{d-1} )} = \| \chi_1 \|_{L^2(\R^{d-1} )}, \text{ and } \| \nabla \chi_n \|_{L^2(\R^{d-1} )} \le C r_n^{-1}.
\end{equation}

We choose spinors in $\dom\, D_\Lambda$ by 
    \begin{equation}\label{eq:psin_eucl}
    \psi_n(s,t) =  \chi_n(s)   \exp(n\ii \langle s, {\xi_0}\rangle-nt) \chi(t) v_0.
    \end{equation}
    We start by estimating their $L^2$-norm using \eqref{eq:est_chi_n}:
    \begin{align}\label{eq:est_psin}
\|\psi_n\|^2_{L^2}  \le \int_{\R^{d-1}}  \chi_n^2(x) \int_0^1 \exp(-2 n t) \di t \di x \le C/n. 
    \end{align}

For the Dirac operator, we use $\nu = - e_d$ and write
$$
D = -\ii c_{\nu} \left(\partial_t + A \right) \quad 
\text{with}
\quad
A  (f \psi)(x) = a(x, \nabla_{\R^{d-1}} f) \psi(x) + f(x) (A \psi)(x).
$$
Here $a(x, \xi)$ is exactly the principal symbol used in 
the Shapiro-Lopatinski condition. (In this case, it does not depend on $x$ but we keep it in the notation for consistency.)
We can now compute
\begin{align*}
    &(\partial_t +A)\psi_n (s,t) \\
    &\quad = n\bigl(-1 + \ii a(s_0, \xi_0) \bigr) \psi_n +\left(\chi'(t) \chi_n(s) + \chi(t) a(s_0, \nabla \chi_n) \right)  e^{n\ii \langle s, {\xi_0}\rangle-nt}v_0,
\end{align*}
where first term vanishes since $a(s_0, \xi_0) v_0 = -\ii v_0 $.
Hence, 
\begin{align*}
      &|(\partial_t +A)\psi_n (s,t)| = e^{-n t}\left|\chi'(t) \chi_n(s) + \chi(t) a(s_0, \nabla \chi_n) \right|.
\end{align*}
Integrating in a similar way as in \eqref{eq:est_psin} and using the bounds in \eqref{eq:est_chi_n} for the $\R^{d-1}$ integrals and $|a(x, \xi)|=|\xi|$ gives
\begin{equation}\label{eq:est_Dpsi_n_eucl}
    \norm{D\psi_n}_{L^2}^2 \le C n^{-1} (1+r_n^{-2}). 
\end{equation}
Recalling that $r_n = n^{-1/2}$, we conclude that
\begin{equation} \label{eq:bound-on-graph-norm}
     \norm{D \psi_n}_{L^2} + \norm{\psi_n}_{L^2} \le C.
\end{equation}

On the other hand, for the $W^{1,2}$-norm, we compute
\begin{align*}
\partial_t \psi_n (s,t)= - n \psi_n(s,t) + \chi_n(s) \chi'(t) e^{n \ii\langle s_0, \xi_0\rangle - n t} v_0.
\end{align*}
The second term decays just as the $L^2$-norm of $\psi_n$, so we obtain
\begin{align} \label{eq:bound-on-H^1-norm}
     \norm{ \psi_n}_{H^1}^2 \ge  \norm{\partial_t \psi_n}_{L^2}^2 \ge C (n -n^{-1}),
\end{align}
which is unbounded. In view of \eqref{eq:bound-on-graph-norm}, the graph norms remain bounded as $n \to \infty$ and hence, the graph norm of the operator does not control the $H^1$-norm. 

\medskip
\noindent\textbf{Step 2 -- general case.}
We will rely heavily on the constructions from \cite[Section 2--4]{BB12}. By using suitable diffeomorphisms, we may identify a tubular neighborhood $U$ of $\partial M$ in $\R$ with the cylinder $\partial M \times [0, t_0)$ for some sufficiently small $t_0$. Sections of $\calS$ with support in $U$ can be identified with $t$-dependent sections of $\calS|_{\partial M}$, $t$ being the distance to $\partial M$, and with this identification, we have
\begin{equation} \label{eq:iso_L_2}
    \int_{M} |\psi|^2 \romd x = \int_0^{t_0}\int_{\partial M \times \{t\}} |\psi|^2(s,t) \romd \omega_{\partial M\times \{t\}} \romd t.
\end{equation}
Here $\romd \omega_{\partial M\times \{t\}}$ is the  volume element induced from the metric $g$ on $M$.

Since $\Lambda_s$ varies smoothly with $s \in \partial M$ there exists $v \in \calC(\calS|_{\partial M})$ such that $v(s_0)= v_0$, $v(s) \in \Lambda_s$, and $\norm{v(s)}_{\calS_s}\le  1$. Without loss of generality, we assume that the support of $v$ is included in $B_{\partial M}(s_0, 1) \subset \partial M$. 
We can then define $\psi \in \dom\, D_\Lambda$ by
 $$
    \psi(s,t) =  \chi(t/t_0) v(s)   . 
    $$
Next, we pick any $\zeta \in C^\infty (\partial M, \R)$ such that $\romd \zeta|_{s_0} = \xi_0$. We can then define 
$$
f_n (s,t) = \exp (\ii n \zeta(s) - n t) \chi_n(s) ,  $$
where $\chi_n$ is, as before, a function with support in $B_{\partial M}(s, r_n=n^{-1/2})$ normalized such that
$$
\norm{\chi_n}_{L^2(\partial M)} = 1, \quad \norm{\romd \chi_n}_{L^2(\partial M)} \le C/r_n.
$$
For $t\in [0,t_0]$ the $L^2$-norm for the induced norms on $\partial M\cong \partial M\times \{0\}$ and $\partial M\times \{t\}$ are equivalent. Thus, in the following we just work with $L^2(\partial M)$ in the estimates.

With these ingredients, we define
$$
\psi_n = f_n \psi.
$$
By \eqref{eq:iso_L_2} we have 
\begin{equation} \label{eq:est_L2_psi_n}
    \norm{\psi_n}_{L^2(M)}^2 \le C \int_0^1 e^{-2 n t} \norm{\chi_n \psi(\cdot, t)}_{L^2(\partial M)}^2 \di t \le C n^{-1}.
\end{equation}

By \cite[Lemma 4.1]{BB12}, in this representation the Dirac operator takes the form 
$$
D= -\ii c_\nu \left( \partial_t + \tilde{A} + R_t \right)
$$
where $\tilde A$ is a first order differential operator on $\partial M$ and $R_t$ satisfies the estimate 
\begin{equation}\label{eq:error_term_bar}
    \norm{R_t \Psi}_{L^2(\partial M)} \le C\left( |t|  \norm{\tilde A \Psi}_{L^2(\partial M)}  +  \norm{ \Psi}_{L^2(\partial M)} \right).
\end{equation}
The principal symbol of $\tilde A$ equals $a_0$.
We compute 
\begin{align}\nonumber
    \tilde A \psi_n &= a_0(\romd f_n) \psi + f_n \tilde A \psi \\
    &= \ii n \, a_0(\romd \zeta)  \psi_n + e^{\ii n \zeta - n t}a_0(\romd \chi_n) \psi+ f_n \tilde A \psi\label{eq:tildeA}
\end{align}
and
\begin{align*}
    \partial_t \psi_n &= - n \psi_n + f_n \partial_t \psi.
\end{align*}
Hence, 
\begin{equation}\label{eq:hence}
    |(\partial_t + \tilde A) \psi_n | \le n|(- 1 + \ii a_0(\romd\zeta)) \psi_n| + C e^{-n t } \left(|\romd \chi_n| + \chi_n \norm{\nabla \psi}_\infty \right).
\end{equation}
Using that all the functions involved are smooth and that  
$$a_0(\romd \zeta) \psi |_{(s_0,t)} = a_0(s_0, \xi_0) \psi(s_0, t) =  -\ii \psi(s_0, t),
$$
we conclude that the first summand on the right hand-side in (\ref{eq:hence})
vanishes at $s_0$ and therefore it can be bounded by $C \operatorname{dist}(s,s_0) \le C r_n$. Hence, we can estimate
\begin{equation}
|\left( \partial_t + \tilde{A}) \right) \psi_n| \le C e^{-n t} \left( n r_n  \chi_n + |\romd \chi_n| +  \norm{\nabla \psi}_\infty \chi_n \right).
\end{equation}
Estimating the $L^2$-norm in a similar way as in \eqref{eq:est_L2_psi_n}, this gives the bound
\begin{equation}
    \label{eq:est_main_Dpsi_n}
    \norm{\left( \partial_t + \widetilde{A}) \right) \psi_n}^2\le C (n r_n^2 +(n r_n)^{-1} + n^{-1}).
\end{equation}
For the contribution of $R_t$, we use that by \eqref{eq:tildeA}
\begin{equation} 
    \norm{\tilde{A} \psi_n( \cdot, t)}_{L^2(\partial M)} \le C e^{-n t}(n + r_n^{-1} + 1). 
\end{equation}
Then, applying \eqref{eq:error_term_bar} and integrating in $t$ gives
\begin{equation*}
    \norm{R_t \psi_n}_{L^2(M)}^2 \le C (n^{-3}(n + r_n^{-1} + 1)^2 + n^{-1} ) \le C n^{-1}.
\end{equation*}
Summing with \eqref{eq:est_main_Dpsi_n} and recalling that $r_n = n^{-1/2}$ finally shows that
\begin{equation}
    \norm{D \psi_n}_{L^2(M)}+ \norm{\psi_n}_{L^2(M)} \le C. 
\end{equation}

For the $H^1$-norm, we use as before 
\begin{equation*}
    |\nabla \psi_n| \ge |\partial_t \psi_n | = |-n \psi_n + f_n \partial_t \psi | \ge n |\psi_n|- f_n \norm{\nabla \psi}_\infty,
\end{equation*}
so that
\begin{equation}
    \norm{\psi_n }_{W^{1,2}}^2 \ge \norm{\partial_t \psi_n}_{L^2(M)}^2 \ge  C(n - n^{-1}).
\end{equation}
\end{proof}

As an application, consider
 now a smooth section $t$ of the unit tangent bundle of $\partial M$. By Lemma~\ref{lem:tangent_Lambda} $\Lambda=E_+(c_{t})$ is a self-adjoint local smooth boundary condition. For this example  we see: 

\begin{corollary}\label{lem:SL_tang} The boundary condition $\Lambda=E_+(c_{t})$ is regular if and only if the boundary has dimension $1$.  \end{corollary}

\begin{proof}Recall that $a(s,\xi)= c_{\nu(s)}c_\xi$. We assume first that the dimension of the boundary is at least two: Let $s\in \partial M$.  Take $\xi  \perp t(s)$ with $|\xi|=1$, then $c_{t(s)}(c_{\nu(s)} c_\xi)=(c_{\nu(s)} c_\xi)c_{t(s)} $. As commuting operators  $c_{t(s)}$ and $c_{\nu(s)} c_\xi$ simultaneously diagonalize.  
	Thus, there is a $v\in E_+(c_{t(s)})\setminus \{0\}$ which is a $+\ii$ or $-\ii$ eigenvector of $c_{\nu(s)} c_\xi$. Hence, $v$ is a $-\ii$  eigenvector of $c_{\nu(s)} c_\xi$ or of $c_{\nu(s)} c_{-\xi}$; i.e., at least one of the sets $E_{-\ii} (a(s,\pm \xi))\cap \Lambda_s$ is not equal to $\{0\}$ and the Shapiro-Lopatinski condition is not fulfilled.  Thus the boundary condition is not regular.\medskip 
	
	Now assume that the boundary is one dimensional: Let $s\in \partial M$ and $k\in T_s(\partial M)$, $|k|=1$. Then  $k=\pm t(s)$ and $a(s,k) = \pm c_{\nu(s)}c_{t(s)}$.  Let $v\in E_+(c_{t(s)})\cap E_{-\ii}(\pm c_{\nu(s)}c_{t(s)})$. Then $-\ii v=\mp c_{\nu(s)}c_{t(s)}v=\mp c_{\nu(s)}v$. But $c_{\nu(s)}^2=\Id$, which implies $v=0$. Thus, $\Lambda=E_+(c_t)$ is Shapiro-Lopatinski and hence regular.
\end{proof}

\begin{remark}\label{rem_noncomp_reg}
For complete manifolds with compact boundary the natural notion of regularity is a local one. To distinguish this from our definition, we call it `locally regular' here, i.e.;\medskip 

The boundary condition is called \emph{locally regular} if $\dom\, D_\Lambda\subset W^{1,2}_{\mathrm{loc}}(\calS)$. The boundary condition is \emph{locally strongly regular} or \emph{locally $\infty$-regular}, see  \cite[Def.~2.3]{BLZ} or \cite[Def.~7.15]{BB12}, if $\{\phi \in \dom\, D_\Lambda\ |\ D^\ell \phi\in L^2_{\mathrm{loc}}(\calS)\ \forall \ell \leq k\}\subset W^{k,2}_{\mathrm{loc}}(\calS)$ for all $k\in \mathbb N_{\geq 0}$.\medskip 

For compact manifolds this definition is equivalent to Definition~\ref{def:regular}.\medskip  

For complete manifolds with compact boundary Propositions~\ref{prop:shap-lop} and~\ref{prop_revSL} remain valid if \emph{(strongly) regular} is replaced by \emph{locally (strongly) regular}. For the second proposition this is true since the support of the sequences constructed to disprove regularity is shrinking to a point and gives hence an argument against local regularity.
\end{remark}

\section{In the presence of a chirality operator}

\label{sec_chiral} 
In this section, we specify to the case where the Clifford bundle $\calS\to M$ admits a chirality operator $\beta\colon \calS\to \calS$, as defined in \S 1 (page~\pageref{def:chiral}). In this situation we can classify all self-adjoint and strongly regular boundary conditions, as will be done in this section.\medskip 

Let $\calS^\pm$ be the subbundles of $\calS$ to the eigenvalues $\pm 1$ of $\beta$. 
Since Clifford multiplication anticommutes with $\beta$ by assumption,
there is a linear map 
$$C\colon T^*M \to \mathrm{Hom} (\calS^+,\calS^-)  \text{ such that } c_\xi = \begin{pmatrix} 0 & C_\xi^*\\ C_\xi &0\end{pmatrix}.$$ 
The anticommutation relation of the Clifford multiplication translates into 
\begin{equation} \label{eq:chiral-anticommutation}
	C_\xi^* C_\eta + C_\eta^* C_\xi= 2 g(\xi, \eta) \Id_{\calS^+} \quad \text{ and } \quad C_\xi C_\eta^* + C_\eta C_\xi^*=2 g(\xi, \eta) \Id_{\calS^-} .
\end{equation}

\subsection{Self-adjointess in the presence of a chirality operator}

We will denote by $U(\calS^+|_{\partial M})$ the subbundle of the endomorphism bundle 
$\text{End}(\calS^+|_{\partial M})$ where each fibre consists of the unitary maps. Analogously, $U(E_+(c_\nu), E_-(c_\nu))$ is the corresponding subbundle of $\mathrm{Hom} (E_+(c_\nu), E_-(c_\nu))$

\begin{proposition} \label{prop:BC_split}
	Let the Clifford bundle over $M$ admit a chirality operator as above. 
	Then, a local smooth boundary condition is self-adjoint if and only if there exists  a section $\widetilde{f}$ of $U(\calS^+|_{\partial M})$ 
	such that for each $s\in\partial M$
	\begin{equation} \label{eq:BC-split}
		\Lambda_s =  \left\{ \left.\begin{pmatrix}
			(\Id + \widetilde{f}(s)) w\\ C_{\nu(s)}(\Id-\widetilde{f}(s))  w 
		\end{pmatrix} \right| w \in \calS_s^+\cong \mathbb C^{N/2}\right\}.
	\end{equation}
\end{proposition}

\begin{proof}
	Since $C_{\nu(s)}^* C_{\nu(s)} = \Id$, the eigenspaces $E_\pm(c_\nu(s))$ of $c_\nu(s)$ associated to the eigenvalues $\pm 1$ are given by 
	$$
	E_\pm (c_\nu(s)) = \left\{ \begin{pmatrix}
		v \\ \pm C_{\nu(s)} v
	\end{pmatrix}   \Big|\ v \in \calS^+_s\cong \bbC^{N/2} \right\}.
	$$
	By Theorem~\ref{thm:local}, a local smooth boundary condition $\Lambda$ is self-adjoint if and only if for each $s \in \partial M$, $\Lambda_s$ is the graph of $F(s)$ for some section $F$ of $U(E_+(c_\nu), E_-(c_\nu))$
	
	For such an $F$, define  $\widetilde{f}\in \mathrm{End}\, ( \calS^+|_{\partial M})$ by $P|_{\calS^+} F P|_{\calS^+}$ , with $P_{\calS^+}$ the orthogonal projection on $\calS^+$. 
	In matrix notation, we have
	\begin{equation} \label{eq_f_tildef} F(s)\binom{v}{C_{\nu(s)}v} = \binom{\widetilde{f}(s)(v)}{-C_{\nu(s)}\widetilde{f}(s)(v) }.\end{equation}
	
	Moreover, the splitting $\calS^+\oplus \calS^-$ is orthogonal since $\beta$ is unitary. 
	Thus, $F(s)$ is unitary if and only if $\widetilde{f}(s)\colon \calS_s^+\to \calS_s^+$ is unitary. Hence, every $\widetilde{f}$ is a smooth section in $U(\calS^+|_{\partial M})$ and gives via \eqref{eq_f_tildef} rise to a unique smooth section $F$ in $U(E_+(c_\nu), E_-(c_\nu))$.\medskip  
	
	For the converse, one can check directly that any $\widetilde{f} \in \calC (U(\calS^+|_{\partial M}))$ defines, by formula\eqref{eq:BC-split} a smooth subbundle $\Lambda$ of rank $N/2$ that satisfies the symmetry condition \eqref{eq:sym_BC}.
\end{proof}

\begin{example} \label{ex:3d_MIT}
In the case $M \subset \R^3$, with $\calS = M \times \C^4$ and the standard Clifford multiplication, $C_\xi = C_\xi^* = \sum_j \sigma_j \xi_j$, where $\sigma_j$ are the Pauli matrices. 
A chirality operator is obtained by taking $\beta = \sigma_3 \otimes \Id$, such that the splitting in $\calS^\pm$ is just the direct sum $\C^4= \C^2 \oplus \C^2$.
In this framework, the \textsc{MIT} boundary condition is obtained by taking $\widetilde{f}(s) = \ii \Id_{2 \times 2}$. By writing $(1+\ii )w = v_1$, $(1-\ii) C_\nu w = v_2$, the boundary condition is the following simple relation between upper and lower components:
$$
v_2 = \frac{1-\ii}{1+\ii} C_\nu v_1 = -\ii C_\nu v_1.
$$
 
\end{example}

\begin{example} \label{ex:generalized_infinite_mass}
    The family of generalized \textsc{mit}-bag conditions studied in \cite{behrndt2020self,rabinovich2021-BVP, arrizabalaga2023eigenvalue} is obtained in a similar way by taking $\widetilde{f}= e^{i\theta}\Id$. To connect to the notations of  \cite{behrndt2020self}, let $\omega = \tan \theta/2$. This family of boundary conditions can be defined in any dimension as soon as the bundle admits a chirality operator $\beta$. 

    In our setting, $\widetilde f$ needs to be a smooth function on the boundary and any such $\widetilde f$ gives rise to a self-adjoint boundary condition.
    In contrast, in \cite[Thm. 5.9]{behrndt2020self} it is proven that if the absolute value of $\omega$ is nowhere one, corresponding to $\theta \notin \mathbb Z \pi$, the corresponding Dirac operator is self-adjoint. This illustrates the advantage of defining the operator with a closed domain instead of defining the domain as a subset of $H^1$.
    Then, the exceptional values of $\omega$ or $\theta$ are recovered when the regularity of the boundary conditions is considered and we do so in Corollary~\ref{cor:gen_MIT} below. 
    
    On the other hand, the approach from \cite{behrndt2020self} has the advantage to yield operators that have elliptic regularity by definition, and extends to Lipschitz-regular functions $\omega$ parametrizing the boundary condition. In this work, we restrict our attention to smooth functions.
 
\end{example}

\begin{remark} The fiber bundle $U(\calS^+|_{\partial M})\to M$ always admits many smooth sections
	(besides the identity map).
	A procedure to construct such $\widetilde{f}$ is as follows: 
	Using local trivializations of $\calS^+|_{\partial M}$ and a partition
	of unity, one can construct smooth maps $H\colon \calS^+|_{\partial M}\to\bbR$ that restrict
	to Hermitian quadratic forms on each fiber.  Then, using the 
	symplectic form
	on each fiber given by the imaginary part of the Hermitian inner product,
	one can form fiber-wise the Hamilton fields of the restrictions $H|_{\calS_s^+}$.
	These fields glue together to
	define smooth vertical fields $\Xi_H$ on the total space of
	$\calS^+|_{\partial M}$, whose time-$t$ maps are unitary endomorphisms of $\calS^+|_{\partial M}$.
\end{remark}

\subsection{Regularity in the presence of a chirality operator}
In the last subsection we classified all self-adjoint local smooth boundary conditions, see Proposition~\ref{prop:BC_split}. In order to decide which of these boundary conditions are (strongly) regular, it suffices to check the Shapiro-Lopatinski condition from Proposition~\ref{prop:shap-lop}.\medskip 

First, we have the following observation:

\begin{lemma}\label{lem_CliffS+}
    Assume that $\calS$ admits a chirality operator with eigenbundles $\calS^{\pm}$. Then $\calS^\pm|_{\partial M}$ are again bundles of Clifford modules where the Clifford multiplication is given by $R\colon k\in T^*\partial M \mapsto R_k\define \ii C_\nu^*C_k$. Moreover, the boundary operator $A_0$ (See \eqref{eq:def_At}) can be written as
    \begin{equation*}
        A_0 =  D_+ \oplus (-C_\nu D_+ C_\nu^*), \quad \text{ with } D_+ \colon \calS^+|_{\partial M} \to \calS^+|_{\partial M}
    \end{equation*}
a Dirac-type operator. In particular, $\rank\, \calS= 2 \rank\, \calS^+ \ge 2\times 2^{\lfloor(d-1)/2 \rfloor  }$.
\end{lemma}
\begin{proof}

We compute $A_0$
at some  $s\in\partial M$: We choose an orthonormal frame $e_i$ near $s$ with $\nu=e_d\perp \partial M$. Then 
		\begin{equation}\label{}
		D = -\ii\sum_{j=1}^d c_{e_j} \nabla_j = 
		-\ii c_{e_d} \Big( \nabla_d  + \underbrace{c_{e_d} \sum_{j<d} c_{e_j} \nabla_j}_{= A_0} \Big). 
	\end{equation}
	Thus, the symbol of $A_0$ viewed as an operator on $\partial M$ is given by 
 \[a_0(s,\xi)=  c_\nu c_\xi=\begin{pmatrix} C_\nu^* C_\xi & 0 \\ 0 & C_\nu C_\xi^* \end{pmatrix}\] 
 for $s\in \partial M$, $\xi\in T_s^* \partial M$ ($C_\xi$ was defined in Section~\ref{sec_chiral}).
Note that the map $k \in T^*\partial M \mapsto R_k := \ii C_\nu^* C_k \in \operatorname{End}(\calS^+|_{\partial M})$ gives a Clifford module of rank $N/2$ over $\partial M$.
To see this, first note that from \eqref{eq:chiral-anticommutation} and $g(\nu, k) = 0$
$$
R_k^* = -\ii C_k^* C_\nu = R_k.
$$
Then, for $k_1, k_2 \in T^*\partial M$, we have
$$
R_{k_1} R_{k_2} = \bigl(-\ii C_{k_1}^* C_\nu \bigr)\bigl(\ii C_\nu^* C_{k_2} \bigr)= C_{k_1}^* C_{k_2} 
$$
and, once more from \eqref{eq:chiral-anticommutation}, we conclude that
$$
R_{k_1} R_{k_2}+ R_{k_2} R_{k_1} = 2 g(k_1,k_2).
$$
Therefore, it suffices to define $D_+ = -\ii \sum_{j= 1}^{d-1} R_{e_j} \nabla_{e_j}$.
\end{proof}

\begin{remark}
We note that $\calS^\pm|_{\partial M}$ together with the Clifford multiplication $R$ from above and the hermitian metric and connection induced from $\calS$ is in general \underline{not} a Clifford bundle:     The compatibility of Clifford multiplication and connection fails in general and a term including the second fundamental form of the boundary appears in \eqref{eq:nabla_Cliff}.
\end{remark}
With the notation of the last lemma, for each $s \in \partial M$ and $\xi \in T^*\partial M$ with $|\xi|= 1$ we define the eigenspaces 
\begin{equation}\label{def_Gpm}
    G_{\pm 1 } (s,\xi) = \{ v \in \calS^+_s\, |\, R_\xi v  = \pm  v \} .
\end{equation} 
Associated to a unitary $\widetilde{f}(s) \in U(\calS^+_s)$, we define the eigenspaces
\begin{equation*}
    F_{\pm 1}(s)=  \{ v \in \calS^+_s\, |\, \tilde f(s) v  = \pm  v \}
\end{equation*}
and collect all the other eigenspaces into 
\begin{equation*}
 F_\perp(s) = (F_{+1}(s) + F_{-1}(s))^\perp.
 \end{equation*}
Finally, we define the Cayley transform of $\widetilde f $ restricted to $F_\perp$ by
\begin{equation}\label{eq_Q}
    Q = -\ii (\Id +\widetilde f)^{-1} (\Id - \widetilde f) P_{F_\perp},
\end{equation}
with $P_{F_\perp}$ the orthogonal projection on $F_\perp$. Since $\widetilde f$ is pointwise unitary, $Q$ is pointwise Hermitian.

\begin{proposition} \label{prop:LS-for-split}
       Assume that the Clifford bundle admits a chirality operator and write the boundary condition in the form \eqref{eq:BC-split} for some $\widetilde{f}\in U(\calS^+|_{\partial M})$.
Define the eigenspaces $G_{\pm 1}$ and $F_{\pm 1}$, $F_\perp $ and the operator $Q$ as before. 
 This boundary condition fulfills the Shapiro-Lopatinski condition at $s$ if and only if 
\[ \left(F_{+1}(s) \oplus  F_{-1}(s) + \begin{pmatrix}
     1 \\ \ii  Q(s) \end{pmatrix}F_\perp(s)\right) \cap (G_{+1}(s, \xi) \oplus G_{-1}(s, \xi))  =\{0\}\]
for all $\xi \in UT^*_s\partial M$.\medskip 

 This in particular implies the necessary conditions for being Shapiro-Lopatinski that 
   \begin{align} \label{eq:LS-for-split-caseA}
        Q G_{+ 1}(s, \xi)\cap G_{ - 1} (s, \xi) &=\{0 \} \\
        \label{eq:LS-for-split-caseB}
        F_{+1}(s) \cap G_{+1}(s,\xi) = F_{-1}(s) \cap G_{- 1}(s,\xi) &=\{0 \}
   \end{align} 
   for all $\xi \in T_s^*(\partial M)$ with $|\xi|_g=1$.

In the special case that $F_{+1}(s)=F_{-1}(s)=\{0\}$,  \eqref{eq:LS-for-split-caseA}
is equivalent to being Shapiro-Lopatinksi.   
\end{proposition}
	
 \begin{proof}

 We fix $s$ and $\xi \in T^*_s\partial M$ with $|\xi| = 1$ and will omit the dependence on $s$ and $\xi$ from the notation.
By the previous proposition, the principal symbol of $A_0$ equals
$$
a(s,\xi) = \bigl(-\ii  R_\xi \bigr) \oplus \bigl( \ii C_\nu R_\xi C_\nu^* \bigr)
$$
Since $a(s,\xi)$ is block-diagonal and $C_\nu^*$ onto, the $- \ii$-eigenspace of $a(s,\xi)$ is given by
	\begin{equation*}
		E_{- \ii} =G_{+ 1} \oplus (C_\nu G_{-1}).
	\end{equation*}

By Proposition~\ref{prop:BC_split} there is a $\widetilde{f}\in U(\calS^+|_{\partial M})$ such that 
  \begin{equation} 
        \Lambda_s =  \left\{ \left.\begin{pmatrix}
     (\Id + \widetilde{f}(s)) w\\ C_{\nu(s)}(\Id-\widetilde{f}(s))  w 
 \end{pmatrix} \right| w \in \calS_s^+\cong \mathbb C^{N/2}\right\}
   \end{equation}
	We decompose any $w \in \calS_s^+$ as $w= w_{+1} + w_{-1} + w_\perp$, where each component belongs to the associated eigenspaces $F_{\pm 1}$ and $F_\perp$ of $\widetilde f$, so that 
\begin{align}\nonumber\Lambda_s
&=\left\{
\begin{pmatrix}
   2 w_{+1} +  (\Id + \widetilde{f}(s)) w_\perp \\ 2 C_{\nu(s)} w_{-1} + (\Id-\widetilde{f}(s))  w_\perp 
 \end{pmatrix}  \Bigg| w_{\star} \in F_{\star}  \right\} 
 \\
&=  F_{+1}(s) \oplus (C_{\nu(s)} F_{-1}(s)) + \begin{pmatrix}
     1 \\ \ii C_{\nu(s)} Q(s)
 \end{pmatrix}  F_\perp (s).\label{eq_decomp_Lambda}
  \end{align}
  Here, the second summand indicates the image of $F_\perp$ through the map $v\mapsto (v, \ii C_\nu Q v)^{\top}$.

Thus, the condition in the proposition is $\Lambda_s\cap E_{-\ii}(s, \xi)=\{0\}$, where the line of second component is multiplied by $C_\nu$. The necessary conditions \eqref{eq:LS-for-split-caseB}
and \eqref{eq:LS-for-split-caseA} then directly follow.\end{proof}

As an application of this result, we obtain the regularity of generalized infinite-mass boundary conditions.

\begin{corollary} \label{cor:gen_MIT}
The generalized infinite mass boundary conditions of example~\ref{ex:generalized_infinite_mass} are regular for $\theta\not\in \pi \mathbb Z$ and otherwise not regular.
\end{corollary}
\begin{proof}
    If $\tilde f = e^{\ii \theta}$ for some $\theta \notin \pi \mathbb Z$,  $F_\perp = \calS^+$ and $Q = -\ii(\Id + \tilde f)^{-1}  (1 -\tilde f) = -\tan (\theta/2) \in \R$. 
    In particular, $Q (G_{-1}(s,\xi)) \cap G_{+1} (s,\xi) = \{0\}$.

    For $\theta \in \pi\mathbb Z$ we have $\tilde{f}=\pm \mathrm{Id}$. Hence $F_{\pm 1} = \calS^+$ and \eqref{eq:LS-for-split-caseB}  is never fulfilled. 
\end{proof}

In the next section, we use these results to find all regular boundary conditions in low dimensions.
\section{Regular self-adjoint local smooth boundary conditions in low dimensions.}

In this section we apply the results from the last sections in low dimensions $d\leq 5$ and low ranks $N\leq 4$ of the Clifford bundle, and we recover results that are known in the Euclidean case. We recall that $N\geq  2^{\lfloor d/2 \rfloor}$ is necessary for the Clifford multiplication to exist. As a warm-up, we treat the planar case $d=2, N=2$. We then move on to odd dimensions $d=3,5$ with $N=2$ (resp. $N=4$), where we show that no regular self-adjoint local smooth boundary conditions exist. Finally, for the cases $d=3,4$ and $N=4$, we find an explicit parametrization for the regular self-adjoint local smooth boundary conditions.

\subsection{The two-dimensional case.}
As a warmup, we study the case $d=2$, $N=2$, and recover known results. In this case we do have a chirality operator and $\calS= \calS^+\oplus \calS^-$ and $\calS^\pm$ are just complex line bundles. We recover the boundary conditions due to \cite{BerryMondragon} for the Euclidean case.

\begin{proposition} Let  $d=2$, $N=2$. Every smooth function $B\colon \partial M \to \R \setminus \{0\} $ defines via 
\begin{equation}
    \Lambda_s = \{(w, \ii B(s)C_{\nu(s)}w)\ |\ w \in \calS^+_s \cong \C\}
\end{equation}
a regular self-adjoint boundary condition. 
   Conversely, every regular self-adjoint local smooth boundary condition is of the above form for a smooth function $B\colon \partial M \to \R\setminus \{0\}$.
  \end{proposition}

In the special case where $M$ is a subset of $\mathbb R^2$ we can choose the Clifford multiplication such that $C_{e_1}=\Id$ and $C_{e_2}=-\ii \Id$ for the standard basis $e_i$ of $\mathbb R^2$. Then the boundary condition can be written as    $\Lambda_s = \{(w, \ii B(s)(t_1(s)+\ii t_2(s)) w\ |\ w \in \calS^+_s \cong \C\}$ where $t$ is a unit tangent vector field of $\partial M$ and $B$ is as above.
  
  \begin{proof} By Proposition~\ref{prop:BC_split} a self-adjoint local smooth boundary condition has the form   \begin{equation} \label{eq:BC-2-2-case} 
        \Lambda_s =  \left\{ \left.\begin{pmatrix}
     (\Id + \widetilde{f}(s)) w\\ C_{\nu(s)}(\Id-\widetilde{f}(s))  w 
 \end{pmatrix} \right| w \in \calS_s^+\cong \mathbb C\right\}
   \end{equation}
   for unitary maps $\widetilde{f}(s)$ depending smoothly on $s$. In this case, $\calS^{\pm}$ have rank $1$. Referring to the notation of Proposition~\ref{prop:LS-for-split}, we conclude that for each $s$ and $\xi$, either 
$G_{+1}(s,\xi) = \calS^+_s$ and $G_{-1} (s,\xi) = \{0\}$, or the opposite holds. Thus, if $F_{-1}$ and $F_{+1}$ are nontrivial, the same applies their intersection with $G_{\pm 1}(s,\xi)$ for a suitable choice of $\xi$. Thus, for a regular boundary condition, $F_\perp = \calS_s^+$. Since for each $\xi$ one of $G_{\pm 1}(s,\xi)$ is trivial, the intersection $Q G_{-1}(s,\xi) \cap G_{+1}(s,\xi)$ is trivial.
Define $B = Q$, then $B \colon \partial M \mapsto \R \setminus \{0\}$. Since $B$ and $C_\nu$ commute, the boundary condition rewrites as \eqref{eq:BC-2-2-case}.
  \end{proof}

    The case $d=2$ and $N=4$ is relevant for the description of graphene. In the Euclidean setting, an explicit parametrization for self-adjoint boundary conditions was given in \cite{AkhmerovBeenakker},
       and their regularity has been studied in \cite{benguria2022block}. It was also shown there that there exist pointwise unitary transformations that allow us to write the boundary conditions in a block-diagonal form and, for this reason, their regularity can be obtained from the case $N=2$.\medskip 

       Fix the chirality operator to be Clifford multiplication by the volume form times $\ii$ which gives again a decomposition $\calS =\calS^+ \oplus \calS^-$. Since $\partial M$ is one-dimensional, $\calS^+|_{\partial M}$ is a trivial $\C^2$-bundle. This isomorphism can always be chosen to be a fibrewise isometry of Hermitian bundles. For the following we assume that $\calS^+|_{\partial M} = \Sigma\times \C^2$.

\begin{proposition}
  Let $d=2$, $N=4$. Fix the chirality operator to be Clifford multiplication by the volume form times $\ii$.
  Any regular self-adjoint local smooth boundary condition is of the form
$$
\Lambda_s =  \left\{ \begin{pmatrix}
   w \\
    i C_{\nu(s)} A(s) w 
\end{pmatrix} \middle|\ w \in \C^2 \right\}
$$
where $A(s)$ is an invertible Hermitian matrix depending smoothly on $s \in \partial M$.  
  \end{proposition} 

The Euclidean case with constant $A= A(s)$ corresponds to the setting of \cite{AkhmerovBeenakker} and \cite{benguria2022block}.
  
\begin{proof}
We use the chirality operator to write a self-adjoint local smooth boundary condition as in Proposition~\ref{prop:BC_split},
\begin{equation*} 
		\Lambda_s =  \left\{ \left.\begin{pmatrix}
			(\Id + \widetilde{f}(s)) w\\ C_{\nu(s)}(\Id-\widetilde{f}(s))  w 
		\end{pmatrix} \right|\ w \in  \C^{2}\right\}
\end{equation*}    
where $\widetilde{f}(s)\colon  \C^2\to \C^2$ is unitary. 
For regularity, note that, by our choice of chirality operator, $R_t = \pm \Id $. Let $t$ be the positively oriented tangent vector. W.l.o.g. let $R_t=\Id$ (the other case is analogous). Thus,
returning to Proposition~\ref{prop:LS-for-split}, this means that for each tangent  vector $\xi = \pm t$,  $G_{\mp 1}(s,\xi=\pm t)=\{0\}$ and $G_{\pm 1}(s,\xi=\pm t)=\C^2$. Hence, $QG_{+1}\cap G_{-1}= \{0\}$. So, \eqref{eq:LS-for-split-caseB} is satisfied if and 
only if $F_{+1}$ and $F_{-1}$ are trivial. 
Thus, we can invert $\Id + \widetilde{f}(s)$ and define $A(s) = -\ii(\Id-\widetilde{f}(s))(\Id + \widetilde{f}(s))^{-1}$, which is Hermitian and invertible. The boundary condition reads now 
\begin{equation*} 
		\Lambda_s =  \left\{ \left.\begin{pmatrix}
			 w\\ i C_{\nu(s)}A (s)  w 
		\end{pmatrix} \right| w \in  \C^{2}\right\}. \qedhere
\end{equation*}   
\end{proof}

\subsection{Non-existence of regular self-adjoint local smooth boundary conditions for \tps{d=3}{d3} and \tps{N=2}{N2} resp. \tps{d=5}{d5} and \tps{N=4}{N4}}\label{sect:NoGo}

First we classify all self-adjoint local smooth boundary conditions for $(d,N)\in \{(3,2), (5,4)\}$ and then we show below that none of them is regular.

\begin{proposition} \label{prop:global_d_odd}
	In the cases $d=3$, $N=2$ and $d=5$, $N=4$, self-adjoint local smooth boundary conditions exist if and only if there exists a smooth section $t$ of the unit tangent bundle. In the case $d=3$, the boundary conditions read
	$$
	\Lambda_s  = E_+(c_{t(s)}), 
	$$
	i.e., $\Lambda$ is the positive eigenbundle of $c_t$.\medskip 
	
	In the case $d= 5$, the boundary conditions are parametrized by a smooth section of the unit tangent space $t$ of $\partial M$ and a function $\tau\in C^\infty(\partial M)$ and take the form
	$$
	\Lambda_s = E_+(H_s) , \quad H \define \exp(\ii \tau c_\nu) c_t = \cos(\tau) c_t + \ii \sin(\tau) c_\nu c_t.
	$$
\end{proposition}
\begin{proof} By Theorem~\ref{thm:local}(iii) we need to find maps $f\in U(E_+(c_\nu), E_-(c_\nu))$. We analyse the condition at first in a point $s\in \partial M$. Thus, we can work on $\calS_s\cong \mathbb C^N$ and we omit the $s$ in the following:  
	
	The starting point is the observation that a unitary map $f$ from $E_+(c_\nu)$ to $E_-(c_\nu)$ corresponds to a matrix $H$ that maps $\C^N$ to $\C^N$ and takes the form
	\begin{equation}\label{eq:matrix_form}
		H =\begin{pmatrix}
			0 & f^{-1} \\ f & 0 
		\end{pmatrix} \colon E_+(c_\nu) \oplus E_-(c_\nu) \mapsto E_+(c_\nu) \oplus E_-(c_\nu).
	\end{equation}
	The matrix $H$ is Hermitian, traceless, unitary, and anticommutes with $c_\nu$. Conversely, any such matrix $H$ induces a suitable unitary map $f$ and a self-adjoint local boundary condition $\Lambda = E_+(H)$. 
	
	\noindent \textbf{Case $d=3$.} 
	We take Clifford multiplication such that $c_k = \sum_j k_j \sigma_j  \define \sigma \cdot k$. 
	A Hermitian traceless matrix takes the form $v \cdot \sigma$ for some $v \in \R^3$. It is unitary if and only if $|v|^2 = 1$. Finally, it anti-commutes with $c_\nu$ if and only if $g(v,\nu)=0$, so we have $v \in T_s(\partial M)$.
	
	\noindent \textbf{Case $d=5$.}
	We use the $\gamma$-matrices defined by
	\begin{equation}\label{eq_gamma_dim5}
		\gamma_j = \begin{pmatrix}
			0 & \sigma_j \\ \sigma_j & 0 
		\end{pmatrix} \text{ for } j= 1, 2, 3,\quad \gamma_4 = \begin{pmatrix}
			\Id \\ 0&-\Id  
		\end{pmatrix}, \gamma_5 = \begin{pmatrix}
			0 & \ii \Id\\ -\ii \Id& 0 
		\end{pmatrix}
	\end{equation}
	with $\sigma_j$ being the (self-adjoint) Pauli matrices. 
	We choose local coordinates such that at the point $s\in \partial M$ we have $c_\nu = \gamma_4$. Then $E_+ (c_\nu)= \left\{ \binom{a}{0}\ | \ a\in \mathbb C^2\right\}$ and $E_- (c_\nu)= \left\{ \binom{0}{a}\ | \ a\in \mathbb C^2\right\}$.  
	Thus, $f\colon E_+(c_\nu) \to E_-(c_\nu)$ is given by a unitary map $\widetilde{f}\colon \mathbb C^2\to \mathbb C^2$. 
	Anticipating what follows, we parametrize a general map in $\C^{2 \times 2}$ as
	$$
	\widetilde{f}=  e^{\ii \tau} \bigl( -\ii\alpha + \sum_j (v^R_j + \ii v^I_j ) \sigma_j \bigr), \text{ for some } \tau \in [0, 2 \pi), \, \alpha \in \R, \text{ and } v^R, v^I \in \R^3.
	$$
	Then $\widetilde{f}$ is unitary if and only if (we use the identity $\sigma_j \sigma_k = \delta_{jk} + \ii \sum_l \epsilon_{jkl}\sigma_l$)
	\begin{align*}
		1 =& \left( -\ii  \alpha +\sum_j(v^R_j + \ii v^I_j ) \sigma_j \right) \left( \ii \alpha +\sum_k (v^R_k - \ii v^I_k ) \sigma_k\right) \\
		=& \alpha^2 - 2  \alpha \sum_j v^I_j \sigma_j +
		|v^R|^2 + |v^I|^2 + 2 \sum_{j,k,l} \epsilon_{j k l}\sigma_l v_j^R v_k^I \\
		=& \alpha^2 + |v^R|^2 + |v^I|^2  + \sum_j \left(-2 \alpha v_j^I + 2 (v^R \times v^I)_j \right) \sigma_j,
	\end{align*}
	which imposes
	\begin{align*}
		\alpha^2 + |v^R|^2 + |v^I|^2  = 1 
		\text{ and } -2 \alpha v^I + 2 (v^R \times v^I) = 0.
	\end{align*}
	If $\alpha \neq 0$, the second condition implies $v^I = 0$, since $(v^R \times v^I)$ is perpendicular to $v^I$. If $\alpha = 0$, we need that $v^R$ and $v^I$ are collinear. In that case, we may redefine $\tau$ and $v_R$ and assume $v^I = 0$ as well.
	
	So, a unitary map is of the form 
	$$
	\widetilde{f}=  e^{\ii \tau} \bigl( \ii\alpha + \sum_j v^R_j \sigma_j \bigr) \text{ with } \alpha^2 + |v^R|^2 = 1.
	$$
	We define a unit tangent vector $t = (v^R_1, v^R_2, v^R_3, 0, \alpha) \in \R^5$. 
	Then, using \eqref{eq:matrix_form}, the corresponding boundary condition is of the form $E_+(H)$ with 
	\begin{align*}
		H 
		&= \begin{pmatrix}
			0 & e^{-\ii \tau} \bigl(  \ii\alpha + \sum_j v^R_j \sigma_j \bigr) \\
			e^{\ii \tau} \bigl( -\ii\alpha + \sum_j v^R_j \sigma_j \bigr) & 0 
		\end{pmatrix} \\
		& = 
		e^{-\ii \tau \gamma_4} \bigl( \alpha \gamma_5 + \sum_j v^R_j \gamma_j\bigr) \\
		&= e^{-\ii \tau c_\nu} c_{t}
	\end{align*}
	This last form is independent of the particular choice of coordinates.
\end{proof}

\begin{proposition}
	In the case $d=3$, $N=2$, there are no regular self-adjoint local smooth boundary conditions.
\end{proposition}
\begin{proof}
	In view of Proposition~\ref{prop:global_d_odd}, a self-adjoint local smooth boundary condition takes the form $\Lambda_s = E_+(c_{t(s)})$ for some unit tangent vector $t(s)$. By Lemma~\ref{lem:SL_tang} this boundary condition is not regular in dimension $3$.
\end{proof}

\begin{proposition}
	In the case $d=5$, $N=4$ there are no regular self-adjoint local smooth boundary conditions.
\end{proposition}
\begin{proof}
	In view of Proposition~\ref{prop:global_d_odd}, a self-adjoint local smooth boundary condition takes the form $\Lambda_s = E_+(H_s)$ with $H=\exp ( \ii\tau c_{\nu}) c_t$ for some unit tangent vector field $t$ on $\partial M$ and $\tau\in C^\infty(\partial M)$.\medskip 
	
	We fix $s\in \partial M$ and pick a basis such that $\nu= e_4$ and $t(s)= e_5$ and we use for the Clifford multiplication the $\gamma$-matrices as in \eqref{eq_gamma_dim5}. We omit $s$ in the notation in the following. Then 
	$$\exp (  \ii\tau c_{\nu})c_{t} = \begin{pmatrix}
		0 &-\ii  e^{-\ii\tau} \\\ii e^{\ii \tau} & 0
	\end{pmatrix}, \text{ and } \Lambda\define E_+(\exp (\ii \tau c_{\nu})c_{t}) = \left\{\begin{pmatrix}
		w \\ \ii e^{\ii\tau} w
	\end{pmatrix}\middle| w \in \C^2 \right\}.
	$$
	On the other hand, for a general unit tangent vector $k= (\tilde k, 0, k_5)$,
	$$
	a(k,s)=  c_\nu c_k =  \begin{pmatrix}
		0 &  \ii k_5 + \sigma\cdot \tilde k  \\  \ii k_5 - \sigma\cdot \tilde k & 0
	\end{pmatrix}.
	$$    
	We pick a particular $k$ such that $k_5 = \sin(\tau)$ and $\tilde k \in \R^3$ such that $|\tilde k| = |\cos(\tau)|$. 
	The matrix $\sigma\cdot \tilde k$ is Hermitian and has eigenvalues $\pm \cos(\tau)$. We pick $w_0 \in \C^2$ normalized and such that $\sigma\cdot \tilde k \, w_0 = -\cos(\tau) w_0$.
	Then $(w_0,  \ii e^{\ii \tau} w_0)^\top$ is a nonzero element of $\Lambda$. Furthermore, 
	\begin{align*}
		a(k,s) \begin{pmatrix}
			w_0 \\ \ii e^{\ii \tau} w_0
		\end{pmatrix} 
		&= 
		\begin{pmatrix}
			0 &  \ii \sin(\tau) + \sigma\cdot \tilde k  \\  \ii \sin(\tau) - \sigma\cdot \tilde k & 0
		\end{pmatrix} 
		\begin{pmatrix}
			w_0 \\ \ii e^{\ii \tau} w_0
		\end{pmatrix} \\
		&=
		\begin{pmatrix}
			\ii e^{\ii \tau}  (\ii \sin(\tau) - \cos(\tau)) w_0\\ (\ii \sin(\tau) + \cos(\tau))w_0
		\end{pmatrix} = - \ii \begin{pmatrix}
			w_0 \\ \ii e^{\ii \tau} w_0
		\end{pmatrix}
	\end{align*}

	Thus, the Shapiro-Lopatinski condition is not satisfied in this case.
\end{proof}

\subsection{\tps{d=3,4}{d34} and \tps{N=4}{N4}}
In the case $d=3$ we assume additionally that there is a chirality operator (For $\R^3$ with four component spinors this is true and for $d=4$ this is automatic as well). Since the Shapiro-Lopatinski condition only involves eigenspaces of matrices in $\C^2$, we can reduce the question about $Q$ (as in \eqref{eq_Q}) into a question about M\"obius transforms in the extended complex plane, as we will see below.\medskip 

For the three-dimensional case, we recall that $\calS^+|_{\partial M}$ is a bundle of Clifford modules  with Clifford multiplication 
\[ R\colon T^*\partial M \to \End {\calS^+}|_{\partial M},\  k\mapsto R_k=\ii C_\nu^*C_k,\] see Lemma~\ref{lem_CliffS+}. Since $\partial M$ is two dimensional, there is an associated chirality operator $\beta_+ \in \End{\calS^+}|_{\partial M}$. 

\begin{proposition} \label{prop:SL_in_d=3-general}
    In the case $d=3$, $N=4$, assume that $\calS$ admits a chirality operator and take $\beta_+$ a chirality operator associated to the Clifford module $\calS^+|_{\partial M}$. A self-adjoint local smooth boundary condition satisfies Shapiro-Lopatinski  at $s\in \partial M$ if and only if it is of one of the following three forms:
    Either
    \begin{equation}
        \label{eq:SL_in_d=3-caseA}
\Lambda_s = \left\{ ( w, \ii C_\nu A_s w)^{\top } \, \middle|\,   w \in \calS^+|_{\partial M} \right\}, \text{ or }
\Lambda_s = \left\{ ( A_s w, \ii C_\nu  w)^{\top } \, \middle|\,  w \in \calS^+|_{\partial M} \right\}
    \end{equation}
    for some non-degenerate Hermitian matrix
    $A_s = \sfa_s  I + \sfd_s \beta_+ + R_{t_s} $ with $t_s \in T_s^*(\partial M)$, $\sfa_s, \sfd_s \in \R$ and such that $|t_s| < |\sfa_s| $,
   or 
     \begin{equation}
            \label{eq:SL_in_d=3-caseB}
\Lambda_s = E_+(b_s \beta_+ +R_{t_s}) \oplus C_\nu E_- (b_s \beta_+ +R_{t_s})
     \end{equation}
        for some vector $t_s \in  T_s(\partial M)$ and $b_s \in [-1,1] \setminus \{0\}$ such that $|t_s|^2 + b_s^2 =1$. 
  \end{proposition}

The four-dimensional case is actually easier to state.
\begin{proposition} \label{prop:SL_in_d=4-new}
    In the case $d=4$, $N=4$, Shapiro-Lopatinski is satisfied at $s \in \partial M$ if the boundary condition takes the form
     $$
      \Lambda_s = \left\{  \begin{pmatrix}
   \ii A_s w \\  C_{\nu(s)} w
\end{pmatrix}   \, \middle|\, w \in \bbC^{2} \right\} \text{ or } \Lambda_s = \left\{  \begin{pmatrix}
   \ii  w \\  C_{\nu(s)} A_s w
\end{pmatrix}  \, \middle|\,  w \in \bbC^{2} \right\}
     $$
with $A_s$ Hermitian, $A_s \neq 0$, and $\det(   A_s) \ge 0 $. 
  \end{proposition}

The remainder of this section contains the proofs of both propositions. 
We start with the following topological lemma.
To state it, we define $\Gr_1(\calS^+_s)$ as the Grassmannian of one-dimensional subspaces of $\calS^+_s$. This manifold is homeomorphic to the sphere $\bbS^2$ or the Riemann sphere $\hat \bbC$. Indeed, after fixing a basis in $\calS^+_s$, each one-dimensional complex subspace is of the form 
\begin{equation}\label{eq:Grassmann_to_Riemann}
    G_\zeta = \{(z, \zeta z)| z \in \C\} \text{ for some } \zeta \in \hat{\bbC}, 
\end{equation}
with the abuse of notation $G_\infty = \{0\} \oplus \C$.  

\begin{lemma}\label{lem:topological_onto}
  Fix $s \in \partial M$ and consider the unit cotangent space $UT^*_s\partial M$. We consider the map $h\colon \xi\in  UT^*_s \partial M \mapsto G_{+1}(s, \xi) \in \Gr_1(\calS^+_s)$, with $G_{+1}(s,\xi)$ the positive eigenspace to $R_\xi$ as defined in \eqref{def_Gpm}.  Then:
  \begin{enumerate}
      \item If $d=4$ and $N=4$, $h$ is onto.
      \item If $d=3$ and $N=4$, take as a basis for $\calS^+_s$ the eigenspaces of the chirality operator $\beta_+|_{\calS^+_s}$ and use this basis to identify $G_{+1}(s,\xi)\in \Gr_1(\calS^+_s)$ with $\zeta \in \hat \bbC$. Then $h$ is onto $\bbS^1$. 
  \end{enumerate}
\end{lemma}
\begin{proof} 
The proofs of both statements are similar. The key point is that $h$ is injective. Indeed, given $k_1, k_2 \in UT_s^*\partial M$ such that $h(k_1) = h(k_2)$, there is $v\neq 0$ such that 
$R_{k_1} v = R_{k_2} v $. But 
$$
\norm{v} = \frac12 \norm{(R_{k_1}R_{k_2}+ R_{k_2}R_{k_1})v } = |g(k_1,k_2)| \norm{v}. 
$$
Hence $k_1 = k_2$ or $k_1 = - k_2$. The latter is not possible since $R_{-k_1} = - R_{k_1}$.  

For the case $d= 4$, note that both $UT_s^*\partial M$ and $\Gr_1(\calS^+_s)$ are homeomorphic to  $\bbS^2$. 

We conclude that $h$ induces a  smooth injection from the two-sphere into itself. If $h$ was not onto, it would imply an embedding from the $2$-sphere in $\R^2$, which contradicts the Borsuk-Ulam theorem.

\medskip
The case $d=3$ is similar but depends on the choice of a basis. 
Luckily, a natural global basis exists since in this case $\partial M$ is two-dimensional and thus, $\calS^+|_{\partial M}$ admits a chirality operator $\beta_+$. 
With respect to this basis, the Clifford multiplication $R$ is off-diagonal, so we can write
\begin{align}\label{eq:R_d3}
R_{k} = \begin{pmatrix}
    0 & \theta_k^* \\ \theta_k & 0
\end{pmatrix} \text{with }\quad  \theta_k \in \bbS^1 \subset \C .
\end{align}
In this basis, we have $G_{+1}(s, k) = \{ (v, \theta_kv)\, |\, v\in \C\}$.
Since $UT_s^*\partial M$ is also  diffeomorphic to $\bbS^1$, we conclude that $h$ induces a  smooth injection from the circle into itself, which implies, again that it is a homeomorphism. 
\end{proof}

The next ingredient is the following fact about M\"{o}bius transforms.
\begin{lemma} \label{lem:mobius}
    Consider the inversion $\calR\colon \zeta\in \hat{\mathbb C} \define \mathbb C\cup \{\infty\} \mapsto -1/\bar\zeta\in \hat{\mathbb C}$ and the M\"obius transform 
    $$
M(\zeta) = \frac{a \zeta + b}{b^*\zeta + d}, \quad a,d \in \R, ad-|b|^2 \neq 0.
    $$
The equation $M(\zeta) = \calR(\zeta)$ has solutions
\begin{enumerate}
    \item on the unit circle if and only if $|b|^2 \ge ((d+a)/2)^2$.
    \item in the extended complex plane $\hat{\mathbb C}$ if and only if $|b|^2 > a d$.
\end{enumerate}
\end{lemma}
\begin{proof}
    We first reduce the question to the case $b \in \R$: If $b=0$, this is true, else, consider $\zeta = b/ |b| z$.
    Then 
    $$\calR(\zeta) = \frac{b}{|b|} \calR(z), \quad \text{ and } M(\zeta) = \frac{b}{|b|}\frac{a z + |b|}{|b|z + d }.$$
    Hence, if $\zeta$ is a solution in $\hat \bbC$ (resp. on the unit circle), then $z$ is a solution of the equation with $b$ replaced by $|b|$ in $\hat \bbC$ (resp. on the unit circle).

    For the first point, if $z$ is a solution on the unit circle, then $\calR(z) = -z$. 
    The equation simply becomes 
    $$
    a z + b =-z(b z +d); \mathrm{\ i.e.\ }  b z^2 + (a+d) z + b = 0.
    $$
    
    If $b$ is zero and $a\neq -d$, the only root is zero which is not on the unit circle. 
   If $b$ is zero and $a=-d$, then the equation is always fulfilled. 
    Let now $b\neq 0$: If $(d+a)^2 - 4 b^2 \le 0 $, there is a pair of complex conjugate roots $z_0, \bar z_0$ (since the coefficients of the quartic are real). We find that $z_0 + \bar z_0 = (a+d)/b$, $z_0 \bar z_0 = 1$. Hence, $z_0$ lies on the unit circle in this case.
On the other hand, if $(a+d)^2 - 4 b^2 >  0$, the roots $z_1, z_2$ are real and distinct. Since also $z_1 z_2 = 1$, they do not lie on the unit circle.

    For the second point, we observe that $M$ maps the real line to itself. If $ad - b^2 \le  0$, there are solutions on the real axis: Indeed, for real numbers $t$ we find $\calR(t) = -1/t$, and the equation becomes
       $$
       a t^2 + 2b t + d = 0.
    $$
    On the other hand, if $ad - b^2 > 0$, $M$ is an automorphism of the upper half plane. Hence, it does not map points in the upper half plane to their antipodal points in the lower halfplane, and vice versa. The only potential solutions are on the real axis. But those solutions exist only if $ad - b^2 \le 0$,
\end{proof}

With this in place, we can characterize the regular self-adjoint local smooth boundary conditions for the cases $N = 4$, $d=3,4$.
\begin{proof}[Proof of Proposition~\ref{prop:SL_in_d=3-general}]
    As before, we fix $s \in \partial M$ and decompose $\calS^+|_{\partial M } $ in the eigenspaces of the chirality operator $\beta_+$. In view of Proposition~\ref{prop:LS-for-split}, we have to study the eigenspaces of the unitary transformation $f \colon \calS_s^+ \mapsto \calS_s^+$. First, assume that $F_{-1}(s) = F_{+1}(s) = \{0\}$. Then, by \eqref{eq_decomp_Lambda} the boundary condition takes the form 
    $$
    \Lambda_s = \{(w, \ii C_{\nu(s)} Q w)\ |\ w \in \calS^+_s\}. 
    $$
    The Shapiro-Lopatinski condition is then satisfied if for all $\xi \in T^*_s\partial M$, $Q G_{+1}(s,\xi) \cap G_{-1}(s,\xi) = \{0\}$. 
We use the correspondence between $\Gr_1(\calS^+_s)$ to $\hat \C$ given by \eqref{eq:Grassmann_to_Riemann}.
The action of $Q= \left(\begin{smallmatrix}
    a & b \\ b^* & d
\end{smallmatrix}\right)$ on the one-dimensional subspaces $G_{+ 1}(s,\xi)$ corresponds to the M\"{o}bius transform $\zeta \mapsto (a \zeta + b)/(d \zeta + b^*)$. On the other hand, if $\zeta(k)$ corresponds to $G_{+1}(s,k)$, $\calR(\zeta(k))$ corresponds to the orthogonal subspace $G_{-1}(s,k)$. 
The intersection $Q G_{+1}(s,k) \cap G_{-1}(s,k)$ is  nontrivial for some choice of $k$ if and only if both subspaces coincide for this $k$, hence
the corresponding points $\zeta(k) \in \hat \bbC$ satisfy
 the equation $M(\zeta(k)) = \mathcal{R}(\zeta(k))$.

  By Lemma~\ref{lem:topological_onto}, $\zeta(k)$ span exactly the unit circle. Hence, the first point of Lemma~\ref{lem:mobius} guarantees that the Shapiro-Lopatinski condition is satisfied if and only if $|b| < |d+a|/2$. 
Defining $(a+d)/2 =: \sfa_s$ and $(a-d)/2 =: \sf d_s$ and noting that (cp. \eqref{eq:R_d3})
$$
\begin{pmatrix}
    0 & b \\ b^* & 0 
\end{pmatrix} = R_{t_s} \text{ for some } t_s \in T^*\partial M \text{ with } |t_s| = |b|
$$
allows to write
$$
Q = \sfa_s \Id + \sfd_s \beta_+ + R_{t_s}.
$$
\medskip
This case corresponds to \eqref{eq:SL_in_d=3-caseA} in Proposition~\ref{prop:SL_in_d=3-general} where $A=Q$ is invertible.\medskip 

It remains to analyse the cases where $F_{+1}$ and/or $F_{-1}$ are nontrivial. If this is the case, $F_\perp$ is at most one-dimensional, so $Q$ is multiplication by a number and condition \eqref{eq:LS-for-split-caseA} is always satisfied.

We have to check that $F_{\pm 1}$ intersect $G_{\pm 1}(s,k)$ in a trivial way, for all choices of $k$ in the unit tangent space. Defining, as previously, $\zeta(k)$ to be the point on the unit circle corresponding to $G_{+ 1}(s,k)$, $\zeta(k)$ span the unit circle as $k$ covers the unit circle in the tangent plane. Therefore for this it is necessary that $F_{\pm 1}$ corresponds to some  $G(\zeta)$ with $|\zeta| \neq 1$.
The projection from $\C^2$ to $G(\zeta)$ is given by 
$$
\frac{1}{1+ |\zeta|^2}
\begin{pmatrix}
    1 & \zeta^* \\ \zeta & |\zeta^2|
\end{pmatrix} =\frac{1}{2} \Id + \frac{1 - |\zeta|^2}{2(1 + |\zeta|^2)} \beta_+  +\frac{1}{2}R_{t_s} . 
$$
where, as before, we have chosen some $t_s$ in the tangent space such that $  R_{t_s}/2 $ equals the off-diagonal part of the previous matrix, such that 
$|t_s| = 2 |\zeta|/(1 + |\zeta|^2)$.  
If we define $b_s =\frac{1 - |\zeta|^2}{(1 + |\zeta|^2)} $, this means that $F_{\pm 1}$ equals $E_{+} (b_s \beta_+ + R_{t_s})$.

If $F_{+1}$ and $F_{-1}$ are both nontrivial, this means that $(1 \pm f)$ are, up to a factor $2$, projections on $F_{\pm 1}$. Thus, 
\begin{equation*}
    \Lambda_s = F_{+1} \oplus C_\nu F_{-1} = E_{+}(b_s \beta_+ + R_{t_s}) \oplus C_\nu E_-(b_s \beta_+ + R_{t_s}).
\end{equation*}
Returning to the definitions of $b_s$ and $t_s$, note that 
$b_s = \frac{1 - |\zeta|^2}{(1 + |\zeta|^2)} \in [-1,1] \setminus \{0\} $ and that $b_s^2 + |t_s|^2 = 1$. This settles the case \eqref{eq:SL_in_d=3-caseB}.

Now, consider the case where $F_{-1}$ and $F_\perp$ are both one-dimensional, we find
\begin{equation*}
    \Lambda_s = \left\{ \begin{pmatrix} w \\ C_\nu q w
            \end{pmatrix} + \begin{pmatrix} v \\ 0
            \end{pmatrix}
     \middle|\ w \in F_\perp, v \in E_+(b_s \beta_+ + R_{t_s}) \right\}
 \end{equation*}
 with $q \in \R \setminus {0}$ the value of $Q$ on $F_\perp$.
In this case, $F_\perp$ is the orthogonal space to $E_-(b_s \beta_+ + R_{t_s})$. Note that in this case, the necessary conditions from Proposition~\ref{prop:LS-for-split} are also sufficient. If $w +v\in G_{+1}(s,k)$ and $q w \in G_{-1}(s,k)$, in particular they are orthogonal. This means that $0 = \inner{w + v}{w}= |w|^2$ and therefore, $w = 0$.
 
Now, we define the matrix $A$ as $q$ times the projection on $E_-(b_s \beta_+ + R_{t_s})$, i.e.,
$$
A= \frac{q}{2} (\Id - b_s \beta_+ - R_{t_s}) 
$$
we find 
\begin{equation*}
    \Lambda_s = \left\{ \begin{pmatrix} w \\ C_\nu A w
            \end{pmatrix} 
     \middle| w \in \calS^+_s \right\}.
 \end{equation*}
We can define $\sfa_s = q/2$, $\sfd_s = -q b_s/2 $ and replace $t_s$ by $- q t_s/2$ to bring this matrix in the form~\eqref{eq:SL_in_d=3-caseA} where $A$ is not invertible.
The  final remaining case where $F_{+1}$ and $F_\perp$ are one-dimensional is completely analogous.
\end{proof}

The proof of Proposition~\ref{prop:SL_in_d=4-new} is similar but easier (using (ii) of Lemma~\ref{lem:mobius}).


\section{Transmission conditions}\label{sec:transm}
Transmission boundary conditions allow to define Dirac operators 
on a manifold divided in two parts by a hypersurface. We first introduce the necessary notation.
Let $(M,g)$ be a Riemannian closed oriented manifold. Let $\Sigma\subset M$ be a closed 
hypersurface in $M$  such that $M\setminus \Sigma = M_1\sqcup M_2$  (disjoint union).
Then $\overline{M_i}\subset M$ for $i=1,2$ is a compact Riemannian manifold with boundary $\Sigma$.

Let $\calS$ be a Clifford bundle over $M$ with associated Dirac operator $D$. We consider $\widehat{M} =\overline{M_1}\sqcup \overline{M_2}$ with the induced metric and Clifford bundle $\widehat{\calS}= \calS|_{\overline{M_1}} \sqcup \calS|_{\overline{M_2}}$. Then, $\partial \widehat{M} = \partial \overline{M_1} \sqcup \partial \overline{M_2} \cong \Sigma \sqcup \Sigma$. Let $\iota_i\colon \Sigma \to \overline{M}_i\subset \widehat{M}$ be the corresponding identifications.

\begin{definition} \label{def:trans_BC} In the situation described above we call a smooth subbundle $\Lambda\subset \calS|_{\Sigma}\oplus \calS|_{\Sigma}$ a \emph{transmission boundary condition} if there are 
$B_i \in \mathrm{End} (\calS|_{\Sigma})$, $i=1,2$, 
with \[\Lambda_s = \{ (B_1(s) \phi, B_2(s) \phi)  \ |\   \phi \in \calS_s\}\] 
for all $s\in  \Sigma$. The associated Dirac operator $D_\Lambda$  is the one with the domain
\[ \mathrm{dom}\, D_\Lambda = \overline{\{ \psi\in \calC(\widehat{\calS})\ |\ (\psi(\iota_1(s)), \psi(\iota_2(s))\in \Lambda_s, \forall s\in \Sigma\}}^{\Vert .\Vert_D}.\]
\end{definition}

\begin{example}\label{ex_transmissI}
One example is $B_1=B_2= \Id$ which is the most standard transmission condition, see e.g. \cite[Example~7.2.8]{BB12}. 
 \end{example}

 The term of transmission boundary condition has a special meaning in the setting of Dirac operators with singular potentials. 
 In the next example we will review this notion and see that this fits into 
 our Definition~\ref{def:trans_BC}.
 
 \begin{example} \label{ex_transmissII} 
Transmission boundary conditions are models for Dirac operators in $\R^d$ with a singular potential supported on a $(d-1)$-dimensional compact submanifold $\Sigma$, such that $\mathbb R^d\setminus \Sigma=\Omega^+ \sqcup \Omega^-$ ($\Omega^+$ the bounded interior of $\Sigma$). 
This setting has a long history in the mathematical physics literature, going back to \cite{DES89} for the case where $\Sigma$ is the sphere.
For a pedagogical survey on the case $d=3$ with the simplest (electrostatic) $\delta$-shell interaction, see \cite{ourmieres2021survey}. 
We refer to \cite{rabinovich2021-two-d,cassano2022RMI} and \cite{Berndt2024triples} for the analysis of general delta-shell potentials in dimensions $d=2$ and $d=3$, respectively, while \cite{rabinovich2022-delta-Rn} considers the $n$-dimensional case.

To make the correspondence with our framework, we take $M_1= \Omega^+$, $M_2 =\Omega^-$, 
define $\widehat{M}=\overline{M_1}\sqcup \overline{M_2}$ as their disjoint union, 
and let $\calS$ be the trivial spinor bundle on these manifolds. 
In this case $\Omega^-$ is unbounded, but this does not cause any problems in the Euclidean case: a straightforward localization argument suffices to show that the regularity of a boundary condition
only depends on the behavior near the compact boundary $\Sigma$.
For a potential $V \in C^\infty (\Sigma, \End{\calS|_{\Sigma}})$, the boundary condition corresponding formally to $D_\R + V \delta_\Sigma$ is the one ensuring that the jump of the normal derivative across the boundary cancels the singularity $V \delta_\Sigma$. It reads
\begin{equation*}
   (u_+, u_-) \in \Lambda  \iff  -\ii c_\nu (u_- - u_+) = \frac{1}{2} V (u_+ + u_-),
\end{equation*}
where we take $\nu$ as the outward normal corresponding to $\Omega^+$.
The most studied delta-shell potentials are the electrostatic case $V = \tau \Id$ and the Lorentz-Scalar case $V = \lambda \beta$, for $\beta$ the chirality operator that exists in $d=2, N=2$ and $d=3, N=4$. The general version studied in \cite{cassano2022RMI, Berndt2024triples} can be parametrized as 
\begin{equation} \label{eq:delta-potential-general}
    V= \eta \Id + \tau \beta + \omega c_\nu + \ii \lambda c_\nu \beta, \quad \eta, \tau, \omega, \lambda \in \R.
\end{equation}

In order to bring these boundary conditions into the framework of Definition~\ref{def:trans_BC} we define
$A_\pm \define \pm \ii c_\nu - V/2$, so that $(u_+, u_-) \in \Lambda$ if and only if $A_+ u_+ + A_- u_- = 0$. 
In the case where $A_\pm$ are invertible, we can choose $B_1 =  - A_+^{-1}$, $B_2 = A_-^{-1}$. In the non-invertible case with $V$ as in \eqref{eq:delta-potential-general},  at each point $p$ of the boundary $\mathrm{im}\, A_+ \cap \mathrm{im}\, A_-=\{0\}$ (see Example~\ref{ex:transm2}). Hence, we can choose $B_1$ the orthogonal projection onto $\ker A_+$ and $B_2$ the orthogonal projection onto $\ker A_-$.

\end{example}

Back to the general setting of this section,
transmission boundary conditions are not, strictly speaking, local boundary conditions 
in the sense of Definition~\ref{def:local_BC}.
Still, our previous analysis applies to this case by first reducing it to the special case where $M_1 = M_2$. The general case then follows, since the properties of the boundary condition only depend on the geometric data on the boundary and (the Clifford multiplication) of its normal vector.

This results in the following theorem.

\begin{theorem}\label{thm:boundary_transmission}
In the situation described at the beginning of this section, 
let $\Lambda$ be a transmission boundary condition defined via $B_i \in \mathrm{End} (\calS|_{\Sigma})$
and let $\nu$ be a smooth unit normal field along $\Sigma$. Then, the associated Dirac operator $D_\Lambda$ is
\begin{enumerate}[(i)]
\item symmetric if and only if $B_1^\dagger c_\nu B_1 =B_2^\dagger c_\nu B_2$ 
\item self-adjoint if and only if $B_1^\dagger c_\nu B_1 =B_2^\dagger c_\nu B_2$ and $\Ker(B_1)\cap \Ker(B_2) = \{0\}$
\item (strongly) regular if and only if for all $s\in \Sigma$ and $k\in UT^*_s\Sigma$
\[ \mathrm{image} (c_{\nu(s)} B_1(s), B_2(s)) \cap (E_{+\ii}(a(s,k)) \oplus E_{+\ii}(a(s,k)))=\{0\},\]
where $(c_\nu B_1, B_2)$ is meant as element in $\mathrm{Hom}\, (\calS|_{\Sigma}, \calS|_\Sigma \oplus \calS|_\Sigma)$
and $a(s,k)$ is the principal symbol of the operator $A_0$ as in \eqref{eq:def_At}.
If $B_1, B_2$ are full rank, this is equivalent to 
\begin{equation}
    \label{eq:VK-transmission_alt}\Ker(a(s,k)(B_2(s) - B_1(s)) -\ii (B_2(s)+B_1(s))) = \{0\},
\end{equation}
and also  equivalent to 
\begin{equation}
    \label{eq:VK-transmission_alt2}
B_2B_1^{-1} E_{-\ii} (a(s,k)) \cap E_{+\ii}(a(s,k))=\{0\}.
\end{equation}
\end{enumerate}
\end{theorem}

\begin{proof}
We start with a general construction that will be mapped to a special case of the theorem: Let $\calS$ be a Clifford bundle over some compact Riemannian manifold $\widetilde{M}$ with boundary and with Dirac operator $D$. We consider $\hat{\calS}=\calS\oplus \calS$. We use on $\hat{\calS}$ the induced fibrewise hermitian metric and the induced connection. If $c_w$ denotes the Clifford multiplication on $\calS$, we set $\tilde{c}_w= \mathrm{diag} (c_w, -c_w)$ as the Clifford multiplication on $\hat{\calS}$. Then $\hat{\calS}$ is again a Clifford bundle over $\calS$ (with twice the rank of $\calS$).  The  associated Dirac operator will be denoted by $\widetilde{D}$ and equals $(D,-D)$.  Note that $\calS\oplus \calS$ is not the splitting from a chirality operator on $\hat{\calS}$, since Clifford multiplication does not interchange the summands. Let $B_i\in \mathrm{End}(\calS|_{\partial \widetilde{M}})$, $i=1,2$. Then $\Lambda= \{ (B_1\phi, B_2\phi)\in \hat{\calS}|_{\partial \widetilde{M}}\ |\ \phi\in \calS|_{\partial \widetilde{M}}\}\subset \hat{\calS}|_{\partial \widetilde{M}}$ defines a smooth subbundle of $\hat{\calS}|_{\partial \widetilde{M}}$. Moreover, it is symmetric if and only if
\[ 0= \int_{\partial \widetilde{M}} \langle \tilde{c}_\nu \binom{B_1\phi}{B_2\phi}, \binom{B_1\psi}{B_2\psi} \rangle = \int_{\partial \widetilde{M}} \langle B_1^\dagger c_\nu B_1\phi , \psi\rangle - \langle  B_2^\dagger c_\nu 
B_2\phi, \psi\rangle \]
for all $\phi, \psi\in \mathcal{C}(\calS|_{\partial \widetilde{M}})$. This is true if and only if $B_1^\dagger c_\nu B_1 =B_2^\dagger c_\nu B_2$. In this case, if additionally $\mathrm{rank}\, \Lambda = \mathrm{rank}\, \calS$, $D_\Lambda$ is self-adjoint by Theorem~\ref{thm:local}. This holds if and only if $\Ker(B_1) \cap \Ker (B_2) = \{0\}$.
\medskip 

Let $-\ii c_\nu ( \partial_t + A_t)$ be the decomposition of $D$ near the boundary as in 
\eqref{eq:def_At}, and let $a(s, k)\in \mathrm{End} (\calS_s)$ be the principal symbol of $A_0$. Then 
the one of $\widetilde{D}$ is  $-\ii\, \mathrm{diag} ( c_\nu ( \partial_t + A_t), -c_\nu ( \partial_t - 
A_t))$. Hence, the induced Dirac-type operator on the boundary is $\mathrm{diag} (A_0, -A_0)$ 
with principal symbol $\widetilde{a}(s, k)= \mathrm{diag} (a(s,k), -a(s,k))$. Thus, $E_{-\ii}(\widetilde{a}(s,k))= E_{-\ii}({a}(s,k))\oplus E_{\ii}(a(s,k))$. Since $\{c_\nu, a(s,k)\}=0$, we have $c_\nu\colon E_{-\ii}({a}(s,k))\to E_{\ii}(a(s,k))$. Hence, the Shapiro-Lopatinski condition translates to $ \mathrm{im} (c_\nu B_1, B_2) \cap E_{\ii}(a(s,k)) \oplus E_{\ii}(a(s,k))=\{0\}$.\medskip 

If $B_1$ and $B_2$ are full rank, 
this condition is equivalent to $w=0$ being the only solution of the system
$$
\begin{cases}
    a(s,k) c_\nu B_1 w = \ii  c_\nu B_1 w \\
      a(s,k) B_2 w = \ii B_2 w.
\end{cases}
$$
Using $a(s,k) c_\nu = - c_\nu a(s,k)$ (recall that $a(s, k) = c_\nu c_k$) and summing and subtracting the equations gives
$$
\begin{cases}
    a(s,k)(B_2-B_1) w = \ii ( B_1 + B_2) w \\
     a(s,k)(B_1+ B_2) w = \ii ( B_2 - B_1) w.
\end{cases}
$$
Since $a(s,k)^2 = -1$, this shows that the second equation follows from the first one and gives condition \eqref{eq:VK-transmission_alt}.\medskip 

Solving the first equation is equivalent to solving $(a(s,k)+\ii)u= (a(s,k)-\ii) B_2B_1^{-1}u$. Since $a(s,k)^2=-\Id$ and $a(s,k)^\dagger = -a(s,k)$,  we have $\mathrm{image} (a(s,k)+\ii)\cap \mathrm{image} (a(s,k)-\ii) = \{0\}$. Thus, there is no solution, and hence we have regularity, if and only if $B_2B_1^{-1} E_{-\ii} (a(s,k)) \cap E_{+\ii}(a(s,k))=\{0\}$.\medskip 

The situation that we have just considered can be mapped to a special case of the theorem where $\Sigma$ divides $M$ into two isometric parts $M_1$ and $M_2$ such that there is an orientation reversing isometry from $\overline{M_1}\to \overline{M_2}$ that  lifts to the Clifford bundle.  \medskip

For the general case we observe that $M_1$ and $M_2$ 
in fact do not need to be isometric because the problem localizes to the boundary.
\end{proof}

\begin{remark}
In the special case $B_2=\Id$, the conditions in the last theorem read as follows:
The boundary condition is
\begin{enumerate}[(i)]
\item symmetric/self-adjoint if and only if $B_1^\dagger c_\nu B_1=c_\nu$.
\item (strongly) regular if and only if $\mathrm{Ker} \{B_1, a(s,k)\} \cap E_{\ii}(a(s,k))=\{0\}$.
\end{enumerate}
In particular, for $B_1=B_2=\Id$ from Example~\ref{ex_transmissI} we recover that the standard transmission boundary condition is self-adjoint and strongly regular.
\end{remark}

\begin{corollary} \label{ex:transm2}
    The delta-shell potentials from Example~\ref{ex_transmissII} define a self-adjoint transmission problem.
     These problems are regular except for the cases where
    \begin{align} \label{eq:transmision_cond_LS}
       \eta^2 - \tau^2 -\omega^2 = (\lambda \pm 2)^2. 
    \end{align}
\end{corollary} 
\begin{proof}
As defined in Example~\ref{ex_transmissII}, we use
$A_\pm \define \pm \ii c_\nu - V/2$, so that $(u_+, u_-) \in \Lambda$ if and only if $A_+ u_+ + A_- u_- = 0$. 
To study the invertibility of $A_\pm$, we define 
$$
C_{\pm} = A_\pm + \eta \Id = (\pm \ii - \omega/2) c_\nu +\eta/2 \Id -\tau/2 \beta - \ii \lambda/2 c_\nu \beta 
$$
and note that, by using that $\beta$, $c_\nu$ and $\beta c_\nu$ anti-commute pairwise, we have
$$
C_\pm A_{\pm} = -\eta^2/4 + \left( (\pm \ii - \omega/2) c_\nu  -\tau/2 \beta - \ii \lambda/2 c_\nu \beta \right)^2 = 1/4 \left( -\eta^2 + (\pm 2\ii - \omega)^2 + \tau^2 +\lambda^2   \right).
$$
For later use, we define $d_{\pm} = 1/4( -\eta^2 + (\pm 2\ii - \omega)^2 + \tau^2 +\lambda^2)$ and note that, for $d_\pm \neq 0$, $A_\pm$ are invertible with respective inverses $d_\pm^{-1}C_\pm $. Note that $d_+ = \bar d_-$, so either both matrices are invertible or both are singular.

\medskip 
We start out with the case where $A_\pm$ are \textbf{invertible}. Then the boundary condition can be rewritten as
$$
\Lambda = \Ker \begin{pmatrix}
    A_+ & A_-
\end{pmatrix}   = \Ran \begin{pmatrix}
    A_+^\dagger \\ A_-^\dagger
\end{pmatrix}^{\perp} = \Ran \begin{pmatrix}
    - A_+^{-1} \\ A_-^{-1}
\end{pmatrix},
$$
so we may take $B_1 =  - A_+^{-1} = -d_+^{-1}C_+$, $B_2 = A_-^{-1} =  d_-^{-1}C_-$. For later use, note that
note that $B_1^\dagger= - B_2$.

Direct calculation gives that $B_1^\dagger c_\nu B_1=B_2^\dagger c_\nu B_2$. Together with the invertibility of $B_1$ and $B_2$ and Theorem~\ref{thm:boundary_transmission} this implies self-adjointness of the Dirac operator.

\medskip 
For regularity, 
$w=0$ needs to be the only solution of the system
$$
\begin{cases}
    a(s,k) c_\nu B_1 w = \ii  c_\nu B_1 w \\
      a(s,k) B_2 w = \ii B_2 w.
\end{cases}
$$

Multiplying the first equation by $c_\nu d_+$, the second by $d_-$ and adding both, this implies 
\[ a(s,k) (C_-+C_+)w = \ii (C_--C_+)w.\]
Now, we recall that $a(s,k)= c_{\nu} c_k$ and compute
\begin{align*}
 C_- + C_+ &= +\eta \Id - \tau \beta - \omega c_\nu + \ii \lambda c_\nu \beta \\
\ii (C_--C_+) &= 2 c_\nu.
\end{align*}
This simplifies the equation to 
\[(C_-+C_+)w = 2 c_k  w.\]
Returning to the splitting in the eigenspaces of the chirality operator as in Section~\ref{sec_chiral}, we find
\[
\begin{pmatrix}
    \eta - \tau & (\omega - \ii\lambda) C_\nu^* - 2 C_k^* \\
    (\omega +\ii \lambda) C_\nu - 2 C_k& \eta + \tau
\end{pmatrix} \begin{pmatrix}
    w_+ \\ w_-
\end{pmatrix} = 0.
\]
This system reduces to an equation for the upper component
$$
\left(\eta -\tau)(\eta + \tau) - |\omega - \ii \lambda|^2 -4 -  2 \ii \lambda (C_\nu^* C_k - C_k^* C_\nu)\right) w_+ =   0,
$$
or
$$
\left(\eta^2 -\tau^2 - \omega^2-\lambda^2 -4 \right) w_+ =    2 \ii \lambda (C_\nu^* C_k - C_k^* C_\nu) w_+ .
$$
The matrix on the right-hand-side is hermitian and squares to
$16 \lambda^2 \Id $, so its eigenvalues are $\pm 4 \lambda$. 
Furthermore, since $C_{-k} = - C_{k}$, both values occur for some choice of unit tangent vector $k$. 
We conclude that the Shapiro-Lopatinski condition is satisfied unless
$$
\eta^2 -\tau^2 - \omega^2 = (\lambda - 2)^2 \text{ or } \eta^2 -\tau^2 - \omega^2 = (\lambda + 2)^2.
$$

\medskip

It remains to study the \textbf{non-invertible case} and observe that $d_+ = 0$ implies that $\omega = 0$ and
$d_+ = d_- = 1/4(-\eta^2 -4 + \tau^2+\lambda^2)$. 
This simplifies the computations and shows that $\tau^2+\lambda^2 >0$.

We return to the notation of Section~\ref{sec_chiral} and decompose $A_\pm$ in the eigenspaces of $\beta$.  This gives
\begin{equation*}
    A_\pm  = \frac{1}{2} \begin{pmatrix}
        \eta + \tau & \ii(\pm 2-\lambda) C_\nu^* \\
       \ii (\pm 2 + \lambda) C_\nu & \eta -\tau
    \end{pmatrix}.
\end{equation*}
Thus, if $d_\pm = 0$, 
\begin{equation*}
    \Ker (A_\pm) = \left\{\begin{pmatrix}
    (\pm 2 -\lambda) \ii \phi
         \\ (\eta + \tau) C_\nu \phi
    \end{pmatrix} \middle| \phi \in \calS^+  \right\}
\end{equation*}
which shows that $\Ker(A_\pm)$ is $N/2$-dimensional and furthermore, since $2-\lambda \neq (-2 -\lambda)$, that $\Ran(A_+ ) \cap \Ran(A_-) = \{0\}$. 
This means that the transmission boundary conditions do not imply any relation between the traces coming from $\Omega_+$ and $\Omega_-$, and the operator is just a direct sum of Dirac operators on $\Omega_\pm$ with respective boundary conditions $\Lambda_+ = \Ker(A_+)$ and $\Lambda_-= \Ker(A_-)$.

To write the conditions in the framework of transmission boundary conditions, we can take $B_1$ and $B_2$ to be the orthogonal projectors on $\Ker(A_+)$ and $\Ker( A_-)$ and check for these $B_i$ the conditions in Proposition~\ref{thm:boundary_transmission}. But since we just have  a direct sum of Dirac operators on $\Omega_\pm$, we can directly use the results of Section~\ref{sec_chiral}. The expression for $\ker\, A_\pm$ from above shows that they are of the form~\ref{prop:BC_split} with $\tilde{f}=a\mathrm{Id}$ where $a=\frac{(\pm 2-\lambda)\ii + (\eta+\tau)}{(\pm 2-\lambda)\ii - (\eta+\tau)}$ for $ (\eta, \tau,\lambda) \neq (\eta, -\eta, \pm 2) $, and $a=1$ for the remaining cases. Those are exactly the generalized infinite mass boundary condition from Example~\ref{ex:generalized_infinite_mass} and Corollary~\ref{cor:gen_MIT}. Hence, the corresponding Dirac operator is regular if and only if $|a|\neq 1$, i.e., if and only if  $(\eta, \tau,\lambda,\omega) = (\eta, \pm \eta, \pm 2, 0)$. 
Within the set of parameters that have $d_+ = 0$ (i.e., $\omega= 0$ and $\eta^2 - \tau^2 = \lambda^2-4$), these are exactly the parameters that satisfy \eqref{eq:transmision_cond_LS}.

\medskip
To sum up, all sets of parameters $(\eta,\tau,\lambda, \omega)$ give rise to self-adjoint operators in our framework, and except for those satisfying \eqref{eq:transmision_cond_LS}, they are regular and in particular, their domains are included in $H^1$.
\end{proof}

\medskip
The authors declare that they have no conflicts of interest that are directly or indirectly related to the research presented in this paper.


\begin{thebibliography}{AMSPV23}

\bibitem[AB08]{AkhmerovBeenakker}
A.~R. Akhmerov and C.~W.~J. Beenakker.
\newblock Boundary conditions for {D}irac fermions on a terminated honeycomb
  lattice.
\newblock {\em Phys. Rev. B}, 77:085423, 2008.

\bibitem[AMSPV23]{arrizabalaga2023eigenvalue}
N.~Arrizabalaga, A.~Mas, T.~Sanz-Perela, and L.~Vega.
\newblock Eigenvalue curves for generalized {MIT} bag models.
\newblock {\em Communications in Mathematical Physics}, 397(1):337--392, 2023.

\bibitem[Arn00]{arnold2000complex}
V.~I. Arnold.
\newblock The complex {Lagrangian} {Grassmannian}.
\newblock {\em Funct. Anal. Appl.}, 34(3):208--210, 2000.

\bibitem[BB12]{BB12}
Ch. B{\"a}r and W.~Ballmann.
\newblock Boundary value problems for elliptic differential operators of first
  order.
\newblock In {\em In memory of C. C. Hsiung. Lectures given at the JDG
  symposium on geometry and topology, Lehigh University, Bethlehem, PA, USA,
  May 28--30, 2010}, pages 1--78. Somerville, MA: International Press, 2012.

\bibitem[BB16]{BB13}
Ch. B\"{a}r and W.~Ballmann.
\newblock Guide to elliptic boundary value problems for {D}irac-type operators.
\newblock In {\em Arbeitstagung {B}onn 2013}, volume 319 of {\em Progr. Math.},
  pages 43--80. Birkh\"{a}user/Springer, Cham, 2016.

\bibitem[BB24]{baer_bandara}
Ch. Baer and L.~Bandara.
\newblock First-order elliptic boundary value problems on manifolds with
  non-compact boundary.
\newblock Preprint, {arXiv}:2401.17784, 2024.

\bibitem[BBLZ09]{BLZ}
B.~Boo{\ss}-Bavnbek, M.~Lesch, and Ch. Zhu.
\newblock The {C}alder\'on projection: new definition and applications.
\newblock {\em J. Geom. Phys.}, 59(7):784--826, 2009.

\bibitem[BBW93]{bookBooss}
B.~Boo{\ss}-Bavnbek and K.~P. Wojciechowski.
\newblock {\em Elliptic boundary problems for {D}irac operators}.
\newblock Mathematics: Theory \& Applications. Birkh\"auser Boston, Inc.,
  Boston, MA, 1993.

\bibitem[BFSV17a]{BFSV2017-g}
R.~D. Benguria, S.~Fournais, E.~Stockmeyer, and H.~{Van Den Bosch}.
\newblock Self-{A}djointness of {T}wo-{D}imensional {D}irac {O}perators on
  {D}omains.
\newblock {\em Ann. Henri Poincar\'e}, 18(4):1371--1383, 2017.

\bibitem[BFSV17b]{BFSV2017-s}
R.~D. Benguria, S.~Fournais, E.~Stockmeyer, and H.~{Van Den Bosch}.
\newblock Spectral gaps of {D}irac operators describing graphene quantum dots.
\newblock {\em Math. Phys. Anal. Geom.}, 20(2):Art. 11, 12, 2017.

\bibitem[BHM20]{behrndt2020self}
J.~Behrndt, M.~Holzmann, and A.~Mas.
\newblock Self-adjoint {D}irac operators on domains in {$\mathbb{R}^3$}.
\newblock In {\em Annales Henri Poincar{\'e}}, volume~21, pages 2681--2735.
  Springer, 2020.

\bibitem[BHSLS24]{Berndt2024triples}
J.~Behrndt, M.~Holzmann, Ch. Stelzer-Landauer, and G.~Stenzel.
\newblock Boundary triples and {W}eyl functions for {D}irac operators with
  singular interactions.
\newblock {\em Reviews in Mathematical Physics}, 36(02):2350036, 2024.

\bibitem[BM87]{BerryMondragon}
M.~V. Berry and R.~J. Mondragon.
\newblock Neutrino billiards: time-reversal symmetry-breaking without magnetic
  fields.
\newblock {\em Proc. Roy. Soc. London Ser. A}, 412(1842):53--74, 1987.

\bibitem[BSVV22]{benguria2022block}
R.~D. Benguria, E.~Stockmeyer, C.~Vallejos, and H.~{Van Den Bosch}.
\newblock A block-diagonal form for four-component operators describing
  graphene quantum dots.
\newblock {\em arXiv preprint arXiv:2211.07568}, 2022.

\bibitem[CL20]{cassano2020self}
B.~Cassano and V.~Lotoreichik.
\newblock Self-adjoint extensions of the two-valley {D}irac operator with
  discontinuous infinite mass boundary conditions.
\newblock {\em Operators and Matrices}, 14(3):667--678, 2020.

\bibitem[CLMT22]{cassano2022RMI}
B.~Cassano, V.~Lotoreichik, A.~Mas, and M.~Tu{\v{s}}ek.
\newblock General $\delta$-shell interactions for the two-dimensional {D}irac
  operator: self-adjointness and approximation.
\newblock {\em Revista Matem{\'a}tica Iberoamericana}, 39(4):1443--1492, 2022.

\bibitem[DE{\v{S}}89]{DES89}
J.~Dittrich, P.~Exner, and P.~{\v{S}}eba.
\newblock {D}irac operators with a spherically symmetric {$\delta$‐shell}
  interaction.
\newblock {\em Journal of Mathematical Physics}, 30(12):2875--2882, 12 1989.

\bibitem[Gin09]{Ginoux}
N.~Ginoux.
\newblock {\em The {D}irac spectrum}, volume 1976 of {\em Lecture Notes in
  Mathematics}.
\newblock Springer-Verlag, Berlin, 2009.

\bibitem[Hol21]{holzmann2021zigzag}
M.~Holzmann.
\newblock A note on the three dimensional {D}irac operator with zigzag type
  boundary conditions.
\newblock {\em Complex Analysis and Operator Theory}, 47(15), 2021.

\bibitem[H{\"{o}}r94]{HormanderIII}
L.~H{\"{o}}rmander.
\newblock {\em The Analysis of Linear Partial Differential Operators III,
  Pseudo-Differential Operators}.
\newblock Classics in Mathematics. Springer Berlin, Heidelberg, 1994.

\bibitem[JT24]{jud2024classifying}
H.~Jud and C.~Tauber.
\newblock Classifying bulk-edge anomalies in the {D}irac {H}amiltonian.
\newblock {\em arXiv preprint arXiv:2403.04465}, 2024.

\bibitem[OBP21]{ourmieres2021survey}
T.~Ourmi{\`e}res-Bonafos and F.~Pizzichillo.
\newblock {D}irac operators and shell interactions: a survey.
\newblock In {\em Mathematical Challenges of Zero-Range Physics: Models,
  Methods, Rigorous Results, Open Problems}, pages 105--131. Springer, 2021.

\bibitem[OBV18]{ourmieres2018strategy}
T.~Ourmi\'{e}res-Bonafos and L.~Vega.
\newblock A strategy for self-adjointness of dirac operators: applications to
  the mit bag model and $\delta$-shell interactions.
\newblock {\em Publ. Math.}, 62(2):397–437, (2018).

\bibitem[Rab21a]{rabinovich2021-BVP}
V.~Rabinovich.
\newblock Boundary value problems for 3d-dirac operators and mit bag model.
\newblock In {\em Operator Theory and Harmonic Analysis: OTHA 2020, Part I--New
  General Trends and Advances of the Theory 10}, pages 479--495. Springer,
  2021.

\bibitem[Rab21b]{rabinovich2021-two-d}
V.~Rabinovich.
\newblock Two-dimensional {D}irac operators with interactions on unbounded
  smooth curves.
\newblock {\em Russian Journal of Mathematical Physics}, 28(4):524--542, 2021.

\bibitem[Rab22]{rabinovich2022-delta-Rn}
V.~Rabinovich.
\newblock {D}irac operators with delta-interactions on smooth hypersurfaces in
  {$\mathbb{R}^n$}.
\newblock {\em Journal of Fourier Analysis and Applications}, 28(2):20, 2022.

\bibitem[Roe98]{Roe}
J.~Roe.
\newblock {\em Elliptic operators, topology and asymptotic methods.}, volume
  395 of {\em Pitman Res. Notes Math. Ser.}
\newblock Harlow: Longman, 2nd ed. edition, 1998.

\end{thebibliography}
\end{document}